\documentclass[11pt,a4paper]{article}

\usepackage{hyperref}
\usepackage{cite}

\usepackage{amsmath}
\usepackage{amsfonts,amsthm,amssymb}

\numberwithin{equation}{section}

\newtheorem{theorem}{Theorem}[section]
\newtheorem{proposition}{Proposition}[section]

\newtheorem{corollary}{Corollary}[section]

{\theoremstyle{remark}
\newtheorem{rem}{Remark}[section]}
\newtheorem{identity}{Identity}
\newenvironment{identitybis}[1]
  {\renewcommand{\theidentity}{\ref{#1}$^\prime$}%
   \addtocounter{identity}{-1}%
   \begin{identity}}
  {\end{identity}}

\newcommand{\tr}{\operatorname{tr}}

\newcommand{\bra}[1]{\langle\,#1\,|}
\newcommand{\ket}[1]{|\,#1\,\rangle}

\newcommand{\moy}[1]{\langle\,#1\,\rangle}

\def\sul{\sum\limits}
\def\pl{\prod\limits}

\def\la{\lambda}

\def\a{\alpha}
\def\b{\beta}
\def\s{\sigma}

\begin{document}

\begin{flushright}
LPENSL-TH-06/15
\end{flushright}

\bigskip

\bigskip

\begin{center}
\textbf{{\Large  On determinant representations of scalar products and form factors in the SoV approach:  the XXX case}}

\vspace{45pt}

\begin{large}

{\bf N.~Kitanine}\footnote[1]{Universit\'{e} de Bourgogne, Institut de Math\'{e}matiques de Bourgogne, UMR 5584 du CNRS, France; Nicolai.Kitanine@u-bourgogne.fr},~~
{\bf J.~M.~Maillet}\footnote[2]{Laboratoire de Physique, UMR 5672
du CNRS, ENS Lyon,  France;
 maillet@ens-lyon.fr},~~
{\bf G. Niccoli}\footnote[3]{Laboratoire de Physique, UMR 5672 du
CNRS, ENS Lyon,  France; giuliano.niccoli@ens-lyon.fr},

{\bf V.~Terras}\footnote[4]{Universit\'e Paris Sud, LPTMS, UMR 8626 du CNRS, 
France; veronique.terras@lptms.u-psud.fr}
\end{large}

\vspace{45pt}

\today

\end{center}

\vspace{45pt}

\begin{abstract}
In the present article we study the form factors of quantum integrable lattice models solvable by the separation of variables (SoV) method. It was recently shown that these models admit universal determinant representations for the scalar products of the so-called separate states (a class which includes in particular all the eigenstates of the transfer matrix). These results permit to obtain simple expressions for the matrix elements of local operators (form factors). However, these representations have been obtained up to now only for the completely inhomogeneous versions of the lattice models considered. In this article we give a simple algebraic procedure to rewrite the scalar products (and hence the form factors) for the SoV related models as Izergin or Slavnov type determinants. This new form leads to simple expressions for the form factors in the homogeneous and thermodynamic limits. To make the presentation of our method clear, we have chosen to explain it first for the simple case of the $XXX$ Heisenberg chain with anti-periodic boundary conditions. We would nevertheless like to stress that the approach presented in this article applies as well to a wide range of models solved in the SoV framework.
\end{abstract}

 \newpage

\section{Introduction}

Quantum integrable systems are ubiquitous in modern theoretical physics appearing both in statistical mechanics and field theory with applications ranging from condensed matter to string theory \cite{Bet31,McCW73L,Bax82L,Gau83L,FadT87L,Mat93L,Sut04L,Fad82,Fad96,Bei12}. They provide unique possibility to obtain non-perturbative and exact results for strongly correlated systems that cannot be obtained by other methods. Besides the computation of spectrum, scattering matrices and partition functions, one of the main challenges in this domain concerns the exact computation of the form factors and correlation functions that connect to measurable physical quantities in these systems. 

The quantum inverse scattering method \cite{SklF78,FadST79,Skl79,KulS82} together with its associated Yang-Baxter \cite{Yan67,Bax82L,Bax71b,Bax71a} and quantum group structures \cite{Jim85,Jim86,KulR81,Dri85,Dri87} provide a powerful framework to tackle such problems. A central object in this approach is played by the so-called quantum monodromy matrix $T(\lambda)$ depending on a continuous complex parameter $\lambda$ whose matrix elements are operators acting on the quantum space of states of the systems of interest $\mathbb{V}$ and satisfy quadratic commutation relations governed by an $R$ matrix solving the Yang-Baxter cubic equation. The main point of the method is the existence of an abelian sub-algebra of operators acting on $\mathbb{V}$,  including the Hamiltonian of the system, generated by the transfer matrix $\mathcal{T}(\lambda)$ constructed algebraically from $T(\lambda)$ and leading, through an expansion in $\lambda$, to a complete (in the sense of characterization of eigenstates) set of conserved operators and (dynamical) symmetries responsible for the integrability of the model at hand. In this framework, the solution of the original spectral problem for the Hamiltonian is embedded in the  $\lambda$-independent resolution of the spectral problem for the transfer matrix $\mathcal{T}(\lambda)$ as a linear operator on $\mathbb{V}$. At this point, within these algebraic settings, several methods can be used to solve this spectral problem. The first methods historically, Bethe ansatz \cite{Bet31} and algebraic Bethe ansatz \cite{SklF78,FadST79,Skl79,KulS82}, have been extensively used with great success in many paradigmatic integrable systems like Heisenberg spin chains or several lattice discretization of integrable field theories like sine-Gordon. It appeared however that for systems lacking an obvious reference state, namely a simple eigenstate of the transfer matrix $\mathcal{T}(\lambda)$ from which the complete space of states $\mathbb{V}$ can be generated by successive actions of matrix elements of the monodromy matrix $T(\lambda)$ on it, a different, somehow more generic, approach should be designed. The quantum separation of variable (SoV) method initiated by Sklyanin \cite{Skl85,Skl95,Skl90,Skl92} provides such a powerful tool for computing the spectrum and the eigenstates of quantum integrable systems, especially for systems lacking such a reference state. Sklyanin pioneering works on SoV have been presented for fundamental examples of integrable quantum lattice models, e.g. the spin 1/2 XXX chains and the Toda chains, after them several contributions have brought to further develop of the SoV method and nowadays several classes of integrable quantum models are proven to admit a description in the framework of the Sklyanin's SoV method \cite{NieWF09,FraSW08,BabBS96,BabBS97,Smi98,Bab04,vonGIPS06,vonGIPST07,vonGIPST08,vonGIPS09,DerKM03b,DerKM03,DerKM01,BytT06,NicT10,Nic10a,Nic11,GroN12,Nic12,Nic13,Nic13a,Nic13b,FalN14,FalKN14,NicT15}.
One very important feature of the SoV method is that it provides not only the equations determining the spectrum of the transfer matrix but also the proof of its completeness and the construction of the corresponding eigenstates; this is in contrast to the Bethe ansatz or algebraic Bethe ansatz approach where the proof of completeness is in general a non trivial task \cite{TarV95,MukTV09}. 

For all these models, one is of course interested to go beyond the knowledge of the spectrum properties so as to reach their dynamical behaviour  through the computation of their form factors (matrix elements of local operators in the eigenstates basis of the transfer matrix $\mathcal{T}(\lambda)$) and correlation functions. Several progresses in this direction have been achieved in the recent years, in particular for solvable lattice models \cite{Smi92L,BogIK93L,JimM95L,KitMT99,KitMT00,KitMST02a,KitMST05b,BooJMST07,KitKMNST07,KitKMNST08,BooJMST09}. As a matter of fact, the computation of form factors and correlation functions needs an explicit representation of the local operators in the eigenstate basis of the transfer matrix $\mathcal{T}(\lambda)$ together with manageable formulas for the computation of the resulting scalar products of states. While the first problem can be solved within the quantum inverse scattering method by essentially expressing local operators in terms of the quantum monodromy matrix entries, hence solving effectively the so-called quantum inverse scattering problem \cite{KitMT99,MaiT00}, the answer to the second question was up to now strongly dependent on the method used to describe the spectrum and eigenstates of the system. 

In the context of algebraic Bethe ansatz the computation of form factors and correlation functions \cite{IzeKMT99,KitMT99,MaiT00,KitMST02a,KitMST05a} was performed in settings  sufficiently explicit   to compute their critical behaviour  \cite{KitKMST09b,KitKMST09c,KitKMST11b,KitKMST11a,KitKMT14} and make explicit contact (at operator level) with conformal field theories \cite{KozM15}.  All this program was finally pushed forward to time and temperature dependent correlation functions \cite{KozMS10u,KozMS11a,KozT11,KitKMST12,DugGK13,DugGK14}, not mentioning clear contact with structure factors accessible in particular through neutron scattering experiments on magnetic crystals \cite{CauHM05,CauM05}. An essential technical feature in all these results was the appearance of rather sophisticated  determinant formulas, in particular for representations of certain partition functions \cite{Ize87},  scalar products of states and form factors \cite{Sla89,KitMT99}, in a form suitable for their analysis in the thermodynamical limit one is usually interested in \cite{KitKMST09a,KitKMST09b,KitKMST09c,KitKMST10u,KitKMST11a,KitKMST11b,KitKMST12}. 

In the SoV approach, this programme is up to now slightly less developed due to peculiar technical features. To explain this point in more detail, let us recall that in the SoV approach one needs first to identify a diagonalizable  operator, say $S(\lambda)$, computed from the monodromy matrix $T(\lambda)$ entries, having simple spectrum, in such a way that the multidimensional spectral problem for the transfer matrix $\mathcal{T}(\lambda)$ separate in the $S(\lambda)$ eigenstate basis into multiple one dimensional spectral problems, hence leading to its solution. In the example of one dimensional lattice models, the space of states $\mathbb{V}$ is realized as a tensor product of local quantum space of states $V_n$ at sites $n$ of the lattice, $n=1, \dots N$, $\mathbb{V} = \otimes_{n=1}^{N} V_n$. In this case, in order to find a suitable operator $S(\lambda)$ with simple spectrum, one is often led to introduce inhomogeneity parameters $\xi_n$ attached to each site $n$ of the lattice and to let them be in generic positions in the complex plane. This has no effect on the integrability properties of the model and this is somehow a standard method used in  integrable systems; namely to solve a given model by embedding it in a larger class of models depending on continuous parameters, here the $\xi_n$, and use these new variables to obtain simpler resolution process.  All computations are then performed for this inhomogeneous integrable model, the homogeneous limit being taken only at the very end on the quantities one is interested in. It has been first shown in \cite{GroMN12} that within such an SoV framework, the scalar product of separate states and even the form factors of local operators admit simple determinant representations written in particular in terms of Baxter $Q$ functions. The key role in these representations is played by a Vandermonde determinant and its various (straightforward) dressing, making the computation of scalar products in the SoV framework much simpler and transparent compared to the one obtained in the algebraic Bethe ansatz settings. Such representations have quite universal character and properties as it has been shown in many following examples \cite{GroMN14,LevNT15,Nic12,Nic13,Nic13a,Nic13b}. Moreover it can be anticipated that such features will generalize to even more complicated systems associated to higher rank quantum groups, making the SoV approach even more appealing. The power and applicability of all this method  however requires that the determinant representations for the quantities like the form factors or correlation functions are simple or explicit enough in terms of the inhomogeneity parameters $\xi_n$ such that the required homogeneous limits are non ambiguous. In particular one would like to avoid  implicit zero over zero expressions in the homogeneous limit that could just spoil the simplicity and benefits obtained within the SoV method. However, it has been observed already in \cite{GroMN12,GroMN14}, that the homogeneous limit of the determinants for scalar products or form factors of local operators obtained by suitable dressings of the Vandermonde determinant do not have obvious homogeneous limit although there is no doubt that such a limit should, by construction, be well defined. 

It is the purpose of the present article to show that the rather universal determinant representations of scalar products of separate states obtained in the framework of SoV can be generically recast in a different form allowing for an obvious homogeneous limit, hence opening the way to use extensively this approach to tackle dynamical properties of quantum integrable systems. More precisely, it will be shown that the dressed Vandermonde determinants obtained in the SoV approach for scalar products are equal to certain Izergin-type determinants, making the above discussed homogeneous limits trivial. These representations are shown in their turn to be equal to generalized Slavnov-type determinants if one of the two separate states is defined by a solution to equations of Bethe ansatz type. All these equalities between determinants are obtained through purely algebraic identities, hence making their appearance quite universal. For simplicity and also for making the correspondence between SoV results and algebraic Bethe ansatz ones explicit, we made the choice to explain the essential features of our results in a very elementary example : the $XXX$ anti-periodic chain, in this case these formulae are reminiscent of those obtained in \cite{Kos120,Kos12} for the $XXX$ periodic chain. 
Recently several authors studying in particular the properties of the Slavnov formula for the XXX model discovered a relation of this type of determinant with Vandermonde type representations \cite{KosM12,FodW12}. There are also some more complicated relations observed for the XXZ chain \cite{Garb14}.
We would like to stress however that the method presented in this paper seems quite universal leading to similar formulas for more complicated cases, like for $XXZ$ or $XYZ$ Heisenberg chains, even in the presence of generic integrable boundaries,  that will be given in a separate publication. 

The article is organized as follows. In section 2 we recall the basics of the SoV method for the anti-periodic $XXX$ model. In section 3 we describe the determinant formulas for scalar products of separate states that include all eigenstates of the transfer matrix. We then derive the main identities and equivalent representations of these scalar products in terms of Izergin-type and generalized Slavnov-type determinants when one of the state is an eigenstate of the transfer matrix. Similar formulas are then obtained for form factors of local operators in section 4. In  section 5 we show the equivalence of this SoV description with the one obtained from algebraic Bethe ansatz, making use of the fact that by an explicit and simple change of basis the anti-periodic $XXX$ model can be recast in a form allowing its resolution by algebraic Bethe ansatz.

\section{The antiperiodic XXX Heisenberg chain in the SoV framework}

In this section we introduce the XXX Heisenberg chain with antiperiodic boundary conditions and recall the solution of the inhomogeneous version of this model by means of Sklyanin's Separation of Variables approach \cite{Skl85,Skl90,Skl92,Skl95}. Within the framework of this approach, the eigenvalues and eigenstates of the transfer matrix are completely characterized by the solutions of a system of discrete equations involving the inhomogeneity parameters of the model, which can be understood as discrete versions of Baxter's famous $T$-$Q$ equation \cite{Bax82L}.
We show here that this SoV discrete characterization of the transfer matrix spectrum and eigenstates can be reformulated in terms of polynomial solutions of the continuous (i.e., functional) $T$-$Q$ equation.
This enables us to obtain a complete description of the spectrum in terms of the solutions of a system of Bethe-type  equations, and to rewrite the corresponding eigenvectors in an ABA-type form, which  we expect to be more convenient for the consideration of both the homogeneous and thermodynamic limits of the model.

\subsection{The antiperiodic XXX Heisenberg chain}

In this paper we consider the XXX Heisenberg chain of spin $1/2$,
\begin{equation}\label{Ham}
H=\sum_{n=1}^N \left(\sigma_n^x\sigma_{n+1}^x+\sigma_n^y\sigma_{n+1}^y+\sigma_n^z\sigma_{n+1}^z-1\right), \qquad
\end{equation}
with antiperiodic boundary conditions,
\begin{equation}
   \sigma_{N+1}^a= \sigma_1^x\,\sigma_{1}^a\,\sigma_1^x,\qquad a=x,y,z.
\end{equation}
The Hamiltonian \eqref{Ham} acts on a $2^N$-dimensional  quantum space $\mathbb{V}=\otimes_{n=1}^N V_n$, with $V_n\simeq\mathbb{C}^2$.
Here and in the following, $\sigma_n^a$, $a=x,y,z$, stand for the Pauli matrices at site $n$ i.e., acting on the local quantum spin space $V_n$.

The $R$-matrix of the model corresponds to the rational solution of the Yang-Baxter equation,
\begin{equation}\label{mat-R}
R(\lambda )= \left( 
\begin{array}{cccc}
\lambda +\eta & 0 & 0 & 0 \\ 
0 & \lambda & \eta & 0 \\ 
0 & \eta & \lambda & 0 \\ 
0 & 0 & 0 & \lambda +\eta
\end{array}
\right) ,
\end{equation}
where $\lambda$ is the so-called spectral parameter and $\eta$ is an arbitrary non-zero complex parameter.
The monodromy matrix of the XXX spin-1/2 chain 
is defined as the following ordered product of $R$-matrices,
\begin{equation}\label{mon}
T_0(\la)=R_{0N}(\la-\xi_N)\dots R_{01}(\la-\xi_1)
=\begin{pmatrix} A(\la)&B(\la)\\
C(\la)&D(\la)\end{pmatrix}_{\! [0]}.
\end{equation}
This is an operator which acts on the tensor product of a two-dimensional auxiliary space $V_0\simeq\mathbb{C}^2$ with the $2^N$-dimensional quantum space $\mathbb{V}$ of the model.
We have used here the standard notation which may label by indices (at least when it is not clear by the context) the space(s) on which the corresponding operator acts in a non-trivial way.
Note that we have introduced in \eqref{mon} a set of inhomogeneity parameters $\xi_j$, $1\le j \le N$, so that $T_0(\lambda)$ \eqref{mon} corresponds in fact to the monodromy matrix of an inhomogeneous generalization of the XXX spin chain. These inhomogeneity parameters are crucial for the SoV study of the model.

The transfer matrix of the model with antiperiodic boundary conditions can be defined as follows
\begin{equation}\label{transfer}
\mathcal{T}(\la)= \tr_0 \left(\sigma_0^x\, T_0(\la)\right)=B(\la)+C(\la).
\end{equation}
This equation defines a one-parameter family of commuting operators. It is important to underline that the transfer matrix is a polynomial in $\la$ of degree $N-1$.
It is also easy to observe\footnote{See the proof of Theorem \ref{th-ffz+} for an explicit derivation of these symmetries.} that the transfer matrix satisfies the following simple symmetries,
\begin{equation}
\label{x-symmetries}
\left[S^x,\mathcal{T}(\la)\right]=0,\qquad \left[\Gamma^x,\mathcal{T}(\la)\right]=0,
\end{equation}
where 
\begin{equation}
\label{x-symmetries1}
S^x=\sum_{n=1}^N \sigma^x_n, \qquad \Gamma^x=\mathop{\otimes}\limits_{n=1}^N\sigma^x_n=(-i)^N\exp\left(\frac{i\pi}2 S^x\right).
\end{equation}
In the homogeneous limit ($\xi _{n}\rightarrow 0$), the
Hamiltonian \eqref{Ham} of the XXX spin quantum chain with antiperiodic boundary conditions is obtained as a logarithmic derivative of the antiperiodic transfer matrix \eqref{transfer}:
\begin{equation}
\left. H=\mathcal{T}(\lambda )^{-1}\frac{d}{d\lambda }\mathcal{T}%
(\lambda )\right\vert _{\lambda =0}.  \label{ARXFham}
\end{equation}
Finally, let us recall the expression for the quantum determinant of the monodromy matrix, which is a central element of the Yang-Baxter algebra,
\begin{equation}\label{qdet}
\det_q T_0(\la)=B(\la)C(\la-\eta)- A(\la)D(\la-\eta)=-a(\la)\, d(\la-\eta),
\end{equation}
with
\begin{equation}\label{a-d}
a(\la)=\prod_{n=1}^N(\la-\xi_n+\eta),\qquad d(\la)=\prod_{n=1}^N(\la-\xi_n).
\end{equation}

\subsection{Separation of variables for the XXX spin chain}

The functional version of the separation of variables solution of the spin-1/2 XXX spin chain leading to the diagonalization of the transfer matrix \eqref{transfer} comes back to the early works of Sklyanin \cite{Skl90,Skl92}. Following \cite{Nic13}, here we briefly recall the construction of the SoV basis of the space of states as well as the characterization of the transfer matrix spectrum and eigenstates which follow from this approach.

Starting from this point we will always suppose that the inhomogeneity parameters $\xi_a$, $1\le a \le N$, are generic, or at least that they satisfy the condition
\begin{equation}
\xi_a\neq \xi_b\pm h\eta \quad \text{for}\quad h\in\{0,1\},\qquad \forall a\neq b,
\label{generic}
\end{equation}
which ensures the validity of the SoV approach (see \cite{Skl90,Skl92}).

The SoV diagonalization of the transfer matrix \eqref{transfer} relies on the construction of a basis of the space of states, that we shall call SoV basis, which for the antiperiodic case can be chosen as the one diagonalizing the action of the operator $D(\lambda)$ of the monodromy matrix. In this basis, the action of the operators $B(\lambda)$ and $C(\lambda)$, and hence of the antiperiodic transfer matrix \eqref{transfer}, happens to be quasi-local. In the case of the XXX spin-1/2 chain, the construction of such a basis can be explicitly realized in the quantum space $\mathbb{V}$ (respectively its dual space) by a multiple action of operators $B(\xi_j)$ on the reference state (respectively   of operators $C(\xi_j)$ on the dual reference state).
In fact, as noticed in \cite{Ter99}, this SoV basis is nothing else that the so-called $F$-basis introduced in \cite{MaiS00}.

Concretely, let us define, for each $N$-tuple $\mathbf{h}\equiv(h_1,\ldots,h_N)$, a state $\ket{\mathbf{h}}$ in the quantum space $\mathbb{V}$ and a state $\bra{\mathbf{h}}$ in the dual quantum space as
\begin{align}
   & \ket{ \mathbf{h} } 
   = \frac 1{V(\{\xi\})}\prod_{n=1}^{N}\left( \frac{B(\xi _{n})}{a(\xi _{n})}\right) ^{\! h_{n}} \ket{0},
   \label{r-h-state} \\
   &\bra{\mathbf{h} }=\frac 1{V(\{\xi\})}\bra{ 0}\prod_{n=1}^{N}\left( \frac{C(\xi _{n})}{d(\xi _{n}-\eta )}\right) ^{\! h_{n}}, \label{l-h-state}
\end{align}
where $\ket{0}$ and $\bra{0}$ stand respectively for the ferromagnetic reference state with all the spins up (right reference state) and for its dual state (left reference state):
\begin{equation}\label{ref-state}
  \ket{0}= \mathop{\otimes}\limits_{n=1}^N
  \left(\begin{array}{c} 1\\ 0\end{array}\right)_{\! [n]},
  \qquad
  \bra{0}= \mathop{\otimes}\limits_{n=1}^N (1,0)_{[n]}.
\end{equation} 
For convenience, the normalization factor in \eqref{r-h-state} and \eqref{l-h-state}  is chosen to be  a Vandermonde determinant of the inhomogeneity parameters of the model:
\begin{equation}\label{VdMxi}
V(\{\xi\})=\prod_{1\leq b<a\leq N}(\xi _{a}-\xi_{b}).
\end{equation}

The condition \eqref{generic} on the inhomogeneity parameters of the model ensures that the operator $D(\lambda)$ is diagonalizable and has simple spectrum.
It can easily be shown that  the $2^N$ states $ \ket{ \mathbf{h} } $ constitute a complete set of eigenstates for this operator:
\begin{equation}\label{eigen-Dr}
D(\la) \ket{ \mathbf{h} }= d_{\mathbf{h}}(\lambda) \ket{\mathbf{h} },
\end{equation}
and similarly,
\begin{equation}\label{eigen-Dl}
 \bra{ \mathbf{h} } D(\lambda )=d_{\mathbf{h}}(\lambda) \bra{ \mathbf{h}}, 
\end{equation}
where
\begin{equation}
d_{\mathbf{h}}(\lambda)=\prod_{n=1}^N(\la-\xi_n+h_n\eta) .
\end{equation}

Moreover, the basis of the quantum space built from the vectors \eqref{r-h-state} and the basis of the dual quantum space built from the vectors \eqref{l-h-state} are orthogonal. More precisely, with the chosen normalization, the scalar product of a state $\ket{\mathbf{h}}$ of the form \eqref{r-h-state} and of a state $\bra{\mathbf{k}}$ of the form \eqref{l-h-state} is given by
\begin{equation}\label{sp-kh}
    \moy{\mathbf{k}\, |\, \mathbf{h} }
    = \frac{ \delta _{\mathbf{k},\mathbf{h}} }{V(\{\xi\})\, V(\{\xi-h\eta\})},
\end{equation}
where $V(\{\xi\})$ stands for the Vandermonde determinant \eqref{VdMxi} associated with the set of inhomogeneity parameters $\xi_j$, $1\le j\le N$, whereas $V(\{\xi-h\eta\})$ stands for the Vandermonde determinant associated with the set of shifted inhomogeneity parameters $\xi_j-h_j\eta$, $1\le j\le N$.
Hence we have the following decomposition  of the identity $\mathbb{I}$ on the quantum space $\mathbb{V}$:
\begin{equation}\label{decomp-Id}
\mathbb{I} =  V(\{\xi\})\sum_{\mathbf{h}\in\{0,1\}^N}   V(\{\xi-h\eta\})\, \ket{\mathbf{h}}\bra{  \mathbf{h}}.  
\end{equation}

From the SoV basis \eqref{r-h-state} or \eqref{l-h-state} we define as in \cite{Nic13} what we call {\em separate states}. These are states which admit a particular factorized form when expressed in these basis, so that the computation of their scalar product will follow quite straightforwardly from \eqref{sp-kh}. More precisely, these are states which can be written in the form
\begin{align}
   \bra{ \alpha } &=\sum_{\mathbf{h}\in\{0,1\}^N}\prod_{a=1}^{N}\alpha(\xi _{a}-h_a\eta)\ 
                                V(\{\xi-h\eta\})\, \bra{\mathbf{h}}, 
                                \label{left_sep-state} \\ 
   \ket{\beta } & =\sum_{\mathbf{h}\in\{0,1\}^N}\prod_{a=1}^{N}\beta(\xi _{a}-h_a\eta)\ 
                                V(\{\xi+h\eta\})\, \ket{\mathbf{h}} ,\label{right_sep-state}
\end{align}
for any function $\alpha$ or $\beta$. 
It is important to mention here that such states are defined, up to a global normalization, only by the $N$ ratios $\a(\xi_j-\eta)/\a(\xi_j)$ (or $\b(\xi_j-\eta)/\b(\xi_j)$), $1\le j\le N$. This means in particular that many different functions $\alpha$ may lead to the same separate state. 

For the effective computation of the scalar product between states of the form \eqref{left_sep-state} and states of the form \eqref{right_sep-state}, it is convenient to rewrite  \eqref{right_sep-state} with the same Vandermonde determinant as in  \eqref{left_sep-state} by means of the property
\begin{equation}
\label{Vandermonde_trick}
(-1)^N\prod_{n=1}^{N}\left(  \prod_{m=1}^{N}\frac{\xi_{n}-\xi_{m}+\eta}{\xi_{n}-\xi_{m}-\eta}\right)^{h_{n}}V(\{\xi-h\eta\})=V(\{\xi+h\eta\}),
\end{equation}
which can be proven by direct computation.
Then the right separate state  \eqref{right_sep-state} can alternatively be written as
\begin{equation}
   \ket{\beta }  =\sum_{\mathbf{h}\in\{0,1\}^N}\prod_{a=1}^{N}\bar\beta(\xi _{a}-h_a\eta)\ 
                                V(\{\xi-h\eta\})\, \ket{\mathbf{h}} ,\label{right_sep_bar}
\end{equation}
where the function $\bar{\beta}$ is such that
\begin{equation}
\label{betabar}
\bar{\beta}(\xi_n)=\beta(\xi_n),\qquad 
\bar{\beta}(\xi_n-\eta)=-\frac{a(\xi_n)}{d(\xi_n-\eta)}\,\beta(\xi_n-\eta).
\end{equation}

The eigenstates of the antiperiodic transfer matrix \eqref{transfer} happen to be particular cases of separate states, associated with some function $Q(\lambda)$ with ratios $Q(\xi_j-\eta)/Q(\xi_j)$, $1\le j\le N$, fixed by the corresponding eigenvalue of the transfer matrix. More precisely, we can formulate the following theorems, which were proved in \cite{Skl92} and in \cite{Nic13} (for a more general XXZ case).

\begin{theorem}
\label{SoV-spectrum}
Let the inhomogeneity parameters satisfy the condition \eqref{generic}.
Then the antiperiodic transfer matrix $\mathcal{T}(\la)$ \eqref{transfer} has a simple spectrum, and a function $\tau(\la)$ is an eigenvalue of $\mathcal{T}(\la)$ if and only if 
\begin{enumerate}
\item it is a polynomial of degree $N-1$,
\item it satisfies the following set of quadratic discrete equations:
\begin{equation}\label{dis-sys}
\tau(\xi_n)\, \tau(\xi_n-\eta)+a(\xi_n)\, d(\xi_n-\eta)=0,\qquad \forall n\in\{1,\dots,N\}.
\end{equation}
\end{enumerate} 
\end{theorem}

\begin{theorem}\label{th-eigenvector}
The left eigenstate of the transfer matrix $\mathcal{T}(\la)$ corresponding to the eigenvalue $\tau(\la)$ is a left separate state $\bra{Q}$ of the form \eqref{left_sep-state} for a function $Q(\la)$ satisfying the following set of conditions for all $n\in\{1,\dots,N\}$:
\begin{gather}
\tau(\xi_n)\, Q(\xi_n)+a(\xi_n)\, Q(\xi_n-\eta)=0,  \label{LeftTQxi}\\
\big(Q(\xi_n),Q(\xi_n-\eta)\big)\not= (0,0).\label{condQ-2}
\end{gather}
Similarly the right eigenstate with eigenvalue $\tau(\lambda)$ is a  right separate  state of the form \eqref{right_sep-state} for the same function $Q(\la)$. 
\end{theorem}

\begin{rem}
It is important to underline that, even if the system of equations \eqref{LeftTQxi} does not uniquely define the function $Q(\la)$, they nevertheless uniquely define the corresponding eigenstate up to a global normalization factor. 
\end{rem}

\begin{rem}\label{rem-eq-dis}
The conditions \eqref{dis-sys} for the transfer matrix eigenvalue $\tau(\lambda)$ imply that a function $Q(\lambda)$ associated with $\tau(\lambda)$ through the system~\eqref{LeftTQxi} satisfies in fact the following larger systems of $2N$ equations:
\begin{multline}\label{T-Q-dis}
  \tau(\xi_n-h_n\eta)\, Q(\xi_n-h_n\eta)=-a(\xi_n-h_n\eta)\, Q(\xi_n-(h_n+1)\eta)\\
  +d(\xi_n-h_n\eta)\, Q(\xi_n-(h_n-1)\eta),\qquad \forall n\in\{1,\dots,N\},\ \forall h_n\in\{0,1\}.
\end{multline}
Conversely it is easy to see that, if one can exhibit two functions $\tau(\lambda)$ and $Q(\lambda)$ satisfying \eqref{T-Q-dis} and \eqref{condQ-2}, then $\tau(\lambda)$ automatically satisfies the conditions \eqref{dis-sys}. 
It means that the set of quadratic equations \eqref{dis-sys} for $\tau(\lambda)$ in Theorem~\ref{SoV-spectrum} can equivalently be replaced by a condition on the existence of a function $Q(\lambda)$ satisfying \eqref{T-Q-dis}-\eqref{condQ-2}.
\end{rem}

\subsection{Baxter functional  $T$-$Q$ equation}

The system of equations \eqref{T-Q-dis} which, due to Remark~\ref{rem-eq-dis}, leads to the characterization of the transfer matrix spectrum and eigenvectors,  is strongly reminiscent  of  Baxter's famous $T$-$Q$ equation \cite{Bax82L}: it appears naturally as a discrete version, evaluated at the inhomogeneity and shifted inhomogeneity parameters only, of the following functional equation:
\begin{equation}\label{hom-Baxter-Eq}
\tau(\lambda ) \,Q(\lambda )=-a(\lambda )\, Q(\lambda -\eta )+d(\lambda )\, Q(\lambda+\eta ). 
\end{equation}
In fact if, for a given function $\tau(\lambda)$, there exists $Q(\lambda)$ solution of \eqref{hom-Baxter-Eq}, then it is obvious that $Q(\lambda)$ satisfies \eqref{T-Q-dis}. The converse is of course not automatically true\footnote{All that we can say in general is that $Q(\lambda)$ satisfies some inhomogeneous version of \eqref{hom-Baxter-Eq} of the form
\begin{equation}\label{inhom-eq}
   \tau(\lambda ) \,Q(\lambda )=-a(\lambda )\, Q(\lambda -\eta )+d(\lambda )\, Q(\lambda+\eta )+F_Q(\lambda),
\end{equation}
with an additional term $F_Q(\lambda)$ vanishing at all the points $\xi_n$ and $\xi_n-\eta$, $n\in\{1,\ldots,N\}$.}.
It would nonetheless be very convenient to be able to completely characterize the transfer matrix spectrum and eigenstates in terms of solutions of the continuous equation \eqref{hom-Baxter-Eq} rather than of the discrete ones \eqref{T-Q-dis}: the entireness condition for $\tau(\lambda)$ could then simply be rewritten in terms of Bethe-type equations for the roots $\la_j$ of the function $Q(\lambda)$, which would enable us to study the homogeneous and thermodynamic limit of our model in a rather standard way.

Hence the whole problem, to pass from the discrete \eqref{T-Q-dis} to the continuous \eqref{hom-Baxter-Eq} picture, is to understand whether there always exists, for each eigenvalue $\tau(\lambda)$ as characterized from the SoV approach,  a solution $Q(\lambda)$ to the functional equation \eqref{hom-Baxter-Eq} which moreover satisfies \eqref{condQ-2}. In addition, one has to be able to characterize the functional form of this solution, keeping in mind that the latter should be relatively independent from the corresponding eigenvalue $\tau(\lambda)$ if one wants to be able to reformulate our characterization of the spectrum and eigenstates in terms of a system of Bethe-type equations.

Depending on the model one considers, this may be a difficult problem\footnote{See for instance the recent papers \cite{NicT15,LevNT15} which tackle this problem for the slightly more complicated XXZ and dynamical 6-vertex antiperiodic models: it appears that in these cases the functional form (and notably the quasi-periodicity properties) of the solutions $Q(\lambda)$ differ from those of $\tau(\lambda)$. The problem seems even more complicated when one considers spin chains with boundaries, so that it has been suggested in \cite{KitMN14} to  reformulate the SoV characterization of the transfer matrix spectrum in terms of solutions of a particular inhomogeneous functional equations of the type \eqref{inhom-eq} rather than of the homogeneous one.} to precisely identify  the class of $Q$-solutions to \eqref{hom-Baxter-Eq} associated with the transfer matrix eigenvalues $\tau(\lambda)$.
However, in the present case, the situation is quite simple since we have to deal with polynomials of degree less than $N$ only. One can therefore formulate the following theorem, which provides an alternative description of the transfer matrix spectrum with respect to Theorem~\ref{SoV-spectrum}.

\begin{theorem}\label{th-T-Q}
Let the inhomogeneity parameters satisfy the condition \eqref{generic}. 
Then, a given function $\tau(\la)$ is an eigenvalue of the antiperiodic transfer matrix \eqref{transfer} if and only if it satisfies the two following conditions:
\begin{enumerate}
\item it is an entire function of $\lambda$,
\item there exists a polynomial $Q$ of some degree  $R\le N$,
\begin{equation}\label{Q-form}
   Q(\lambda )=\prod_{a=1}^R( \lambda -\lambda_a),
\end{equation}
for some set of roots $\lambda_1,\ldots,\lambda_R$ such that $ \la_a\neq \xi_b,\ \forall a\in\{1,\dots ,R\},\  \forall b \in\{1,\dots ,N\}$,
satisfying with $\tau(\lambda)$ the functional equation \eqref{hom-Baxter-Eq}.
\end{enumerate}
Whenever it exists, such a polynomial $Q(\lambda)$ is unique.
\end{theorem}

\begin{proof}
It is quite straightforward to show that, if the conditions {\it 1} and {\it 2} are satisfied, then $\tau(\lambda)$ is an an eigenvalue of the antiperiodic transfer matrix: on the one hand it is obvious from the previous discussion that the condition {\it 2} of Theorem~\ref{SoV-spectrum} is satisfied; on the other hand the fact that $Q(\lambda)$ is a polynomial of maximal degree $N$ implies that $-a(\lambda)\, Q(\lambda-\eta)+d(\lambda)\, Q(\lambda)$ is a polynomial of maximal degree $2N-1$ so that, from the entireness condition of $\tau(\lambda)$,  the condition {\it 1} of Theorem~\ref{SoV-spectrum} is also satisfied.

Let us now show  that, if the equation \eqref{hom-Baxter-Eq}  admits a polynomial solution $Q(\lambda) $ of the form \eqref{Q-form} for a given function $\tau(\lambda )$, then this solution is unique.
Indeed, let us assume that there exists two different polynomial solutions $P(\lambda )$ and $Q(\lambda )$ to the equation \eqref{hom-Baxter-Eq} associated with the same function $\tau(\lambda)$.  It means that
\begin{equation}
\frac{-a(\lambda )P(\lambda -\eta )+d(\lambda )P(\lambda +\eta )}{P(\lambda )}
=\frac{-a(\lambda )Q(\lambda -\eta )+d(\lambda )Q(\lambda +\eta)}{Q(\lambda )},
\end{equation}
so that
\begin{equation}
   -a(\lambda )\,W_{P,Q}(\lambda )=d(\lambda )\, W_{P,Q}(\lambda +\eta ),
\end{equation}
where $W_{P,Q}(\lambda )$ stands for the quantum Wronskian of these two solutions:
\begin{equation}
W_{P,Q}(\lambda )=Q(\lambda )\, P(\lambda -\eta )-P(\lambda )\, Q(\lambda -\eta ).
\end{equation}
Taking into account that $a(\lambda )=d(\lambda +\eta )$, and using the fact that $W_{P,Q}(\lambda )$ is a polynomial in $\lambda $, we obtain that
\begin{equation}
W_{P,Q}(\lambda )=w_{P,Q}(\lambda )\, d(\lambda )
\end{equation}
where $w_{P,Q}(\lambda )$ is a polynomial in $\lambda $ which moreover has to satisfy
the following quasi-periodicity condition:
\begin{equation}
w_{P,Q}(\lambda +\eta )=-w_{P,Q}(\lambda ).
\end{equation}
Since the only polynomial in $\lambda $ which is periodic of period $2\eta $ is a constant, and since the only constant which satisfies $w_{P,Q}=-w_{P,Q}$ is zero, one has that $W_{P,Q}(\lambda )=0$. Hence $Q(\lambda )=P(\lambda )$ once we have chosen  the highest coefficients of both polynomials  to be equal to 1.

Let us finally prove that, for any eigenvalue $\tau(\lambda )$ of the antiperiodic transfer matrix \eqref{transfer}, there exists  a polynomial solution $Q(\lambda )$ of degree $R\leq  N$ and which does not vanish at the points $\xi_j$, $1\le j \le N$, to the functional equation \eqref{hom-Baxter-Eq} associated with $\tau(\lambda)$.
We recall that, if $Q(\la)$ is a polynomial of maximal degree $N$, then $\tau(\lambda )\, Q(\lambda )$ and $-a(\lambda )\,Q(\lambda -\eta )+d(\lambda )\,Q(\lambda +\eta )$ are both polynomials of maximal degree $2 N-1$.
Hence these two polynomials are equal if and only if they coincide for $2N$ different values of $\lambda$, for instance at the $2N$ points $\xi_a$ and $\xi_a-\eta$ for $a\in\{1,\ldots, N\}$, i.e., if and only the following system of $2N$ equations is satisfied:
\begin{equation}\label{2Nconditions}
\left\{
\begin{aligned}
&\tau(\xi _{a})\,Q(\xi _{a})=-a(\xi _{a})\,Q(\xi _{a}-\eta ),\\
&\tau(\xi_{a}-\eta )\,Q(\xi _{a}-\eta )=d(\xi _{a}-\eta )\,Q(\xi _{a}),
\end{aligned}
\right.
\qquad \forall \, a\in\{1,\ldots, N\}.
\end{equation}
Note at this point that, due to \eqref{dis-sys}, the above system is in fact equivalent to the system of only $N$ equations given by the first line of \eqref{2Nconditions}.
Moreover, saying that $Q(\lambda)$ is a polynomial of maximal degree $N$ is equivalent to saying that $Q(\lambda)$ can be written in the following form: 
\begin{equation}
Q(\lambda )=\sum_{a=1}^{ {N}+1}\prod_{\substack{ b=1  \\ b\neq a}}^{{N}+1}
\frac{\lambda -\xi _{b}}{\xi _{a}-\xi _{b}}\,Q(\xi _{a}).
\label{Interpolation-Q}
\end{equation}
In \eqref{Interpolation-Q}, $\xi_{N+1}$ is an arbitrary complex number, different from $\xi_1,\ldots,\xi_N$, which can be chosen at our convenience.
Hence the system \eqref{2Nconditions} is equivalent to a homogeneous linear system of $N$ equations for the $2N+1$ unknowns $Q(\xi_1),\ldots,$ $Q(\xi_{N+1})$, which can alternatively be thought of as an inhomogeneous linear system for the $N$ unknowns $Q(\xi_1),\ldots,Q(\xi_{N})$ in terms of the $(N+1)$-th one $Q(\xi_{N+1})$:
\begin{equation}\label{system2}
   \sum_{b=1}^N [ c_\tau(\xi_{N+1})]_{ab}\, Q(\xi_b)
   =-\prod_{\ell=1}^N\frac{\xi_a-\xi_\ell-\eta}{\xi_{N+1}-\xi_\ell}\, Q(\xi_{N+1}).
\end{equation}
The elements of the matrix $ c_\tau(\xi_{N+1})$ of this linear system are
\begin{equation}\label{mat-c}
  [ c_\tau(\xi_{N+1})]_{ab}= \delta _{ab}\,\frac{\tau(\xi _{a})}{a(\xi _{a})}
  +\prod_{\substack{ c=1  \\ c\neq a}}^{N+1}\frac{\xi _{b}-\xi _{c}-\eta }{\xi _{a}-\xi _{c}}
  \qquad\forall a,b\in \{1,\ldots,N\}.
\end{equation}
The determinant of this matrix is a rational function of $\xi_{N+1}$ which is not identically zero, hence it is possible to chose $\xi_{N+1}$ for this determinant to be finite and non-zero.
Then, for any given choice of $Q(\xi _{ {N}+1})\neq 0$, there exists one and only one nontrivial solution $\big(Q(\xi_1),\ldots ,Q(\xi_N)\big)$ of the  system \eqref{system2}, which is given by Cramer's rule:
\begin{equation}
   Q(\xi_j)=Q(\xi_{N+1})\, \frac{\det_N\big[ c_\tau^{(j)}(\xi_{N+1})\big] }{\det_N [ c_\tau(\xi_{N+1})]},
   \qquad\forall\, j\in \{1,\ldots, {N}\},
\end{equation}
with matrices $c_\tau^{(j)}(\xi_{N+1})$ defined as
\begin{equation}\label{mat-cj}
   \big[ c_\tau^{(j)}(\xi_{N+1})\big]_{ab}
   = (1-\delta _{b,j}) [ c_\tau(\xi_{N+1})]_{ab}
   -\delta _{b,j} \prod_{\ell=1}^{ {N}}\frac{\xi_a-\xi_\ell-\eta }{\xi _{N+1}-\xi _\ell},
\end{equation}
for all $a,b\in\{1,\ldots, N\}$. Note that the determinant of the matrices are also non-zero rational functions of $\xi_{N+1}$, so that it is also possible to fix $\xi_{N+1}$ such that none of them vanish.
In that case one is ensured that $Q(\xi_j)\not=0$, $\forall\, j\in\{1,\ldots,N\}$.
Hence from  \eqref{Interpolation-Q}  we have obtained a polynomial of degree at most $N$ satisfying the functional equation \eqref{hom-Baxter-Eq} and the requirement that its zeros do not coincide with any of the inhomogeneity parameters. It is important to remark that we have no criteria implying  that the degree of $Q(\lambda )$ is exactly $N$: being described by an interpolation formula in $N+1$ points,  it can be  a polynomial of any degree $ {R}\leq  {N}$.
\end{proof}

Hence, from Theorem~\ref{th-T-Q}, we know that there exists a bijection which relates the set of the transfer matrix eigenvalues $\tau(\lambda)$ to the set of polynomials $Q(\lambda)$ of the form \eqref{Q-form} with maximal degree $N$: the image of $\tau(\lambda)$ is provided by the unique solution $Q(\lambda)$, in the aforementioned set of polynomials, to the functional equation \eqref{hom-Baxter-Eq} associated with $\tau(\lambda)$. From now on we shall denote this unique solution by $Q_\tau(\lambda)$, and we shall denote the corresponding right and left eigenstate by $\ket{Q_\tau}$ and $\bra{Q_\tau}$ respectively. 

The Baxter equation as usual leads to the Bethe equations for the roots of the polynomial $Q(\la)$:
\begin{equation}
\label{Bethe_anti}
\frac {a(\la_a)}{d(\la_a)}\prod_{b=1}^R\frac{\la_a-\la_b-\eta}{\la_a-\la_b+\eta}=1,
\qquad \forall\, a\in\{1,\ldots,R\},
\qquad 1\le R\le N.
\end{equation}
Theorem~\ref{th-T-Q} therefore ensures us that the solutions to these Bethe equations provide a complete description of the spectrum (and eigenstates) of the antiperiodic transfer matrix \eqref{transfer}.

Before concluding this subsection, we would like to mention that it is possible to obtain an alternative (but equivalent) description of this spectrum, as stated in the following theorem.

\begin{theorem}
\label{secondQ}
Let $M=N/2$ if $N$ is even or $M=(N-1)/2$ if $N$ is odd.
Then $\tau(\la)$ is an eigenvalue of the antiperiodic transfer matrix \eqref{transfer} if and only if it can be written in the following form,
\begin{equation}
 \tau(\lambda)=\pm \frac 12 \big[ p(\lambda -\eta )\, q(\lambda +\eta )
                          -q(\lambda -\eta )\, p(\lambda +\eta)\big],
\label{Def-St}
\end{equation}
with $q(\la)$  a polynomial of degree $R\le M$ and $p(\la)$ a polynomial of degree $N-R$ such that
\begin{equation}
   \frac 12 \big[ p(\lambda )\, q(\lambda -\eta )+q(\lambda)\, p(\lambda -\eta )\big]  =d(\lambda ).
\label{Wronskian_pq}
\end{equation}
\end{theorem}

\begin{proof}[Proof]
Let $\tau(\lambda )$ be an eigenvalue of the transfer matrix.
Then Theorem~\ref{th-T-Q} implies that there exists a (unique) polynomial $Q_{\tau}(\lambda )$
satisfying the Baxter equation \eqref{hom-Baxter-Eq} together with  $\tau(\lambda )$.
From  Theorem \ref{SoV-spectrum},  it is easy to see that $ -\tau(\lambda )$ is another eigenvalue of the transfer matrix, so that there also exists a polynomial $Q_{-\tau}(\lambda)$ which satisfies  \eqref{hom-Baxter-Eq} together with  $-\tau(\lambda )$.
We then define  $q(\lambda )\in\{Q_{\tau}(\lambda ),Q_{-\tau}(\lambda )\}$ to be the polynomial with the smaller degree and  $p(\lambda )\in\{Q_{\tau}(\lambda ),Q_{-\tau}(\lambda )\}$ to be the other one.
If the two polynomials $Q_{\tau}(\lambda )$ and $Q_{-\tau}(\lambda )$ have the same degree, we fix for instance $q(\lambda )=Q_{\tau}(\lambda )$ and $p(\lambda )=Q_{-\tau}(\lambda )$.

Using the fact that $q(\lambda )$ and $p(\lambda )$ satisfy the Baxter equation with opposite eigenfunctions, we obtain the identity
\begin{equation}
\frac{-a(\lambda )\, q(\lambda -\eta )+d(\lambda )\, q(\lambda +\eta )}{q(\lambda )}
=\frac{a(\lambda )\, p(\lambda -\eta )-d(\lambda )\, p(\lambda +\eta )}{p(\lambda)},
\end{equation}
or equivalently,
\begin{equation}\label{wpq-id}
a(\lambda )\, \hat{W}_{q,p}(\lambda )=d(\lambda )\, \hat{W}_{q,p}(\lambda +\eta ),
\end{equation}
where we have defined the function $\hat{W}_{q,p}(\lambda )$ as
\begin{equation}
   \hat{W}_{q,p}(\lambda )= \frac 12 \big[ p(\lambda )\, q(\lambda -\eta )+q(\lambda)\, p(\lambda -\eta )\big].
\end{equation}
Since $\hat{W}_{q,p}(\lambda )$ is a polynomial in $\lambda$ and since $a(\lambda )=d(\lambda +\eta )$, the equation \eqref{wpq-id} is satisfied if and only if 
\begin{equation}
 \hat{W}_{p,q}(\lambda )=d(\lambda ).  \label{W-d-form}
\end{equation}
It also means that the degree of $q(\lambda)$ is less or equal to $M$. 

Let us now show that the transfer matrix eigenvalue $\tau(\lambda )$ is of the form \eqref{Def-St}.
By definition we know that
\begin{equation}
\tau(\lambda )\, q(\lambda )=\epsilon [-a(\lambda )\, q(\lambda -\eta )+d(\lambda)\, q(\lambda +\eta )]
\end{equation}
where $\epsilon =1$ if $q(\lambda )=Q_{\tau}(\lambda )$ and $\epsilon =-1$ if  $q(\lambda )=Q_{-\tau}(\lambda )$. Using now that  $a(\lambda )=\hat{W}_{p,q}(\lambda +\eta )$ and that $d(\lambda )=\hat{W}_{p,q}(\lambda )$, we obtain
\begin{equation}
\tau(\lambda )\, q(\lambda )=\epsilon \big[ q(\lambda )\, p(\lambda -\eta )\, q(\lambda +\eta)-q(\lambda )\, p(\lambda +\eta )\, q(\lambda -\eta )\big],
\end{equation}
which implies \eqref{Def-St}.

Vice versa, let $q(\lambda)$ and $p(\la)$ be two polynomials of degree $R$ and $N-R$ respectively and which satisfy \eqref{Wronskian_pq}. Let $\tau^{(+)}(\lambda )$ and $\tau^{(-)}(\lambda)$ be equal to the right hand side of \eqref{Def-St} with + or $-$ sign respectively. We can first remark that, from the definition \eqref{Def-St}, $\tau^{(\pm)}(\lambda)$ are obviously polynomials in $\lambda$ of degree $N-1$. Moreover,
\begin{align}
\tau^{(\pm )}(\lambda )\, q(\lambda ) 
=&\pm \frac 12 \big[ p(\lambda -\eta )\, q(\lambda +\eta)-q(\lambda -\eta )\, p(\lambda +\eta ) \big]\, q(\lambda)  \notag \\
 =&\pm \bigg\{  
       \frac 12 \big[ p(\lambda -\eta )\, q(\lambda)+p(\lambda)\, q(\lambda-\eta)\big]\, q(\lambda +\eta )
       \notag\\
   &\quad\ -\frac 12 \big[ p(\lambda +\eta)\, q(\lambda)+q(\lambda+\eta)\, p(\lambda)\big]\, q(\lambda -\eta )\bigg\}\notag\\
 =&\pm \big[d(\lambda )\, q(\lambda +\eta ) -a(\lambda )\, q(\lambda -\eta )\big]. \label{eq-q}
\end{align}
Similarly one can show that
\begin{equation}\label{eq-p}
   \tau^{(\pm )}(\lambda )\, p(\lambda ) =\mp\big[-a(\lambda)\, p(\lambda-\eta)-d(\lambda)\, p(\lambda+\eta)\big].
\end{equation}
Hence $\tau^{(+)}(\lambda )=-\tau^{(-)}(\lambda)$ satisfies the functional equation \eqref{hom-Baxter-Eq}  with $q(\lambda )$ whereas $\tau^{(-)}(\lambda )=-\tau^{(+)}(\lambda)$  satisfies the functional equation \eqref{hom-Baxter-Eq}  with $p(\lambda )$.
Finally it is easy to see that, given any inhomogeneity parameter $\xi_j$, $p(\lambda)$ and $q(\lambda)$ cannot both vanish in $\xi_j-\eta$: this would imply from \eqref{Wronskian_pq} that $d(\xi_j-\eta)=0$, which is obviously not true. Hence, taking also into account \eqref{eq-q} and \eqref{eq-p}, this means  that $\tau^{(+)}(\lambda )$ and $\tau^{(-)}(\lambda )$ both satisfy \eqref{dis-sys}, so that they are both transfer matrix eigenvalues.
\end{proof}

\subsection{ABA-type representations for the transfer matrix eigenvectors}

In this section we present an alternative way to write the separate states, and hence the eigenstates of the antiperiodic transfer matrix, in a form which is strongly reminiscent of the form of the Bethe states as obtained in the framework of the algebraic Bethe ansatz. It is worth remarking that this type of rewriting can be in fact derived for a large class of integrable quantum models solvable by SoV method.

Let us first define a simple separate state, the state $\bra{ 1 }$ associated with the constant function $\a(\la)=1$.  The corresponding right separate state $\ket{ 1} $  is defined similarly using \eqref{right_sep-state}. Then, if $\alpha(\lambda)$ is a polynomial, the separate state $\bra{\alpha}$ or $\ket{\alpha}$ can be obtained by multiple action of the ``creation'' operator $D$, evaluated at the roots $\alpha_j$ of $\alpha$, on the separate states  $\bra{ 1 }$ or  $\ket{ 1} $.

\begin{proposition}\label{prop-sep-ABA}
Let $\a(\la)$ be a polynomial of the form,
\begin{equation}
 \a(\la)=\prod_{k=1}^R (\la-\a_k).
\end{equation}
Then the corresponding separate states $\bra{\alpha}$ and $\ket{\alpha}$ can be written as
\begin{equation}
\label{ABA_like}
   \bra{\alpha}=(-1)^{RN}\bra{ 1}\prod_{k=1}^R D(\a_k),\qquad
   \ket{\alpha}=(-1)^{RN}\prod_{k=1}^R D(\a_k) \ket{1}.
\end{equation}
\end{proposition}

\begin{proof}
The proof is straightforward using \eqref{eigen-Dr} and \eqref{eigen-Dl}.
\end{proof}

The representation \eqref{ABA_like} can be used notably for the eigenstates  $\ket{Q_\tau}$ and $\bra{Q_\tau}$ corresponding to the eigenvalue $\tau(\lambda)$ of the transfer matrix, by means of the unique polynomial $Q_\tau(\lambda)$ which, in virtue of Theorem~\ref{th-T-Q}, solves the functional equation associated with $\tau(\lambda)$.  Since the roots of $Q_\tau(\lambda)$ can be obtained as the solutions to the corresponding Bethe-type equations, the analogy with algebraic Bethe ansatz is then particularly obvious.
It is also possible to represent these eigenstates in a slightly different form. Let us to this aim define some other simple left and right separate state that we shall denote by $\bra{1_\mathrm{alt}}$  and $\ket{1_\mathrm{alt}}$ respectively, and which are given by a function $1_\mathrm{alt}(\la)$ which alternates sign between the inhomogeneity and shifted inhomogeneity parameters, namely
\begin{equation}
       1_\mathrm{alt}(\xi_a)=1,\quad 1_\mathrm{alt}(\xi_a-\eta)=-1,\quad�\forall a\in\{1,\dots N\}.
\end{equation}
This second representation for the eigenstates   $\ket{Q_\tau}$ and $\bra{Q_\tau}$ corresponding to the eigenvalue $\tau(\lambda)$ of the transfer matrix then uses the polynomial $Q_{-\tau}(\lambda)$ which solves the functional equation \eqref{hom-Baxter-Eq} for the eigenvalue $-\tau(\lambda)$ of the transfer matrix:

\begin{proposition}\label{prop-eigen-ABA}
Let $\tau(\la)$ be an eigenvalue of the antiperiodic transfer matrix \eqref{transfer} and let
\begin{equation}
   Q_\tau(\la)=\prod_{k=1}^R(\la-\la_k),\qquad Q_{-\tau}(\la)=\prod_{k=1}^{N-R}(\la-\widehat{\la}_k),
\end{equation}
be the unique solutions to the functional $T$-$Q$ equation \eqref{hom-Baxter-Eq} associated with the eigenvalues $\tau(\lambda)$ and $-\tau(\lambda)$ respectively.
Then the transfer matrix eigenstate $\ket{Q_\tau}$ with eigenvalue $\tau(\lambda)$ can be represented in following forms:
\begin{align}
    \ket{Q_\tau} &=(-1)^{RN}\prod_{k=1}^R D(\la_k)\, \ket{1} \label{Qtau-form1}\\
    &=(-1)^{(N-R)N}\frac{\prod\limits_{k=1}^R d(\la_a)}{\pl_{k=1}^{N-R} d(\widehat{\la}_a)}\pl_{k=1}^{N-R} D(\widehat{\la}_k)\, \ket{1_\mathrm{alt}}. \label{Qtau-form2}
\end{align}
\end{proposition}

\section{Scalar products of separate states}

In this section, we explain how to compute the scalar products of separate states of the form \eqref{left_sep-state} and \eqref{right_sep-state}. In general, the latter can by construction be represented in terms of the determinant of a weighted sum of two Vandermonde matrices involving the inhomogeneity parameters of the model. We shall notably focus on the case where one of the two separate states is an eigenstate of the antiperiodic transfer matrix. We shall see that in this case it is possible to represent the corresponding scalar product in a more convenient form for the study of the homogeneous and thermodynamic limit, namely, similarly as what happens in the ABA framework for the scalar product of an on-shell and an off-shell Bethe states, in terms of the determinant  introduced in \cite{Sla89} (Slavnov determinant).

\subsection{General determinant representation for the scalar product of two separate states}

Let us start by recalling the general determinant representation for the scalar products of two separate states. Note that this type of representation is a direct consequence of the factorized form of these states in the SoV basis, and hence can be shown for a large variety of models solvable by SoV.

\begin{theorem}
Let $\bra{\alpha}$ be a left separate state of the form \eqref{left_sep-state}, and let $\ket{\beta}$ be a right separate state of the form \eqref{right_sep-state}.
Their scalar product can be written as
\begin{equation}\label{sp-sep}
\moy{ \alpha\, |\, \beta } =\frac{\det_{N}\big[\mathcal{M}^{(\alpha ,\beta)}\big]}{V(\{\xi\})},
\end{equation}
where $\mathcal{M}^{(\alpha,\beta)}$ corresponds to the following weighted sum of two Vandermonde matrices:
\begin{align}
\big[ \mathcal{M}^{( \alpha,\beta ) }\big]_{a,b}
&= \xi _{a} ^{b-1}\, \alpha (\xi _{a})\, \bar{\beta}(\xi _{a})
+(\xi _{a}-\eta) ^{b-1}\, \alpha (\xi _{a}-\eta)\,\bar{\beta}(\xi _{a}-\eta) \label{sp-det-sep1}\\
& =\a(\xi_a)\, \b(\xi_a)\left[\xi_a^{b-1}
  -\frac{a(\xi_a)\, \a(\xi_a-\eta)\,\b(\xi_a-\eta)}{d(\xi_a-\eta)\,\a(\xi_a)\,\b(\xi_a)}\, (\xi_a-\eta)^{b-1}\right].
 \label{sp-det-sep2}
\end{align}
\end{theorem}

The proof of this theorem is straightforward. It was given in \cite{Nic13} in the case of a more general XXZ model.

Note that it is possible to use this determinant representation to show the (expected) orthogonality of the eigenstates of the antiperiodic transfer matrix, which therefore form an orthogonal basis of the quantum space of states of the model.

\begin{corollary}
Let $\tau(\la)$ and $\tau'(\la)$ be two different eigenvalues of the antiperiodic transfer matrix \eqref{transfer}.
Then the corresponding eigenstates are orthogonal:
 \begin{equation}
 \moy{Q_{\tau'}\, |\,Q_\tau}= \langle \,Q_\tau\,\ket{Q_{\tau'}}=0.
 \end{equation}
\end{corollary}

\begin{proof}
$\tau(\la)$ and $\tau'(\la)$ are two different polynomials of degree $N-1$:
\begin{equation}
    \tau(\la)=\sul_{b=1}^N c_b\la^{b-1},\qquad \tau'(\la)=\sul_{b=1}^N c'_b\la^{b-1}.
\end{equation}
It means that the vector with components $v_b=c_b-c'_b$ is non trivial, and it is easy to check that
\begin{equation}
    \sum_{b=1}^N\big[\mathcal{M}^{(Q_\tau,Q_{\tau'})}\big]_{a,b} \, v_b=0.
\end{equation}
\end{proof}

The aim of this section is to show that it is possible to rewrite the expression \eqref{sp-sep} for the scalar product of two separate states into some more convenient forms for the consideration of the homogeneous and thermodynamic limit. We shall in particular link this formula to some (generalizations of some) other determinant representations that have already appeared  in the literature in the framework of algebraic Bethe ansatz: the {\em Izergin determinant},   first introduced as a representation for the partition function of the six-vertex model with domain wall boundary conditions \cite{Ize87}, and the {\em Slavnov determinant}, which appears in the ABA framework as a convenient representation for the scalar product of an on-shell and an off-shell Bethe vectors \cite{Sla89}.
To this aim, we shall first derive some identities that will be used for these reformulations.
 
\subsection{Some useful identities}
 
As announced above, we want to transform the representation \eqref{sp-sep} into a more convenient form for our purpose.
To this aim we introduce, following \cite{KosM12}, some convenient notations that we shall use throughout all this section.
For any set of complex numbers $\{x\}\equiv\{x_1,\dots ,x_M\}$, and a function $f$, we define
\begin{align}
E_{\{ x\} }^{\pm }(y)=&\prod_{n=1}^{N}\frac{y-x_{n}\pm \eta}{y-x_{n}},
\label{DefEpm}\\
\mathcal{A}_{\{ x\}}^{\pm }[ f ] =&\frac{\det_{M}\big[ x_{a}^{b-1}-f(x_{a})\,( x_{a}\pm\eta) ^{b-1}\big] }{V(\{x\}) }
\label{DefApm}
\end{align}
Note that the representation \eqref{sp-sep}-\eqref{sp-det-sep2} of the scalar product can easily be rewritten by means of the notations \eqref{DefEpm} and \eqref{DefApm}: it is enough to choose the functions $\alpha(\lambda)$ and $\beta(\lambda)$ to be some polynomials (which by interpolation is always possible) that we express in terms of their roots $\{\alpha_1,\ldots,\alpha_R\}$ and $\{\beta_1,\ldots,\beta_S\}$ as
\begin{equation}
\alpha (\lambda )=\prod_{n=1}^{R}(\lambda -\alpha _{n}),\qquad
\beta(\lambda )=\prod_{m=1}^{S}(\lambda -\beta _{m}),
\end{equation}
to obtain that
\begin{equation}\label{sp-rep2}
\moy{ \alpha\, |\, \beta }
 =\prod_{n=1}^N\Big(\a(\xi_n)\b(\xi_n)\Big)\ 
 \mathcal{A}^-_{\{\xi\}}\!\left[
             \frac{E^+_{\{\xi\}}}{E^-_{\{\xi\}}}\, E^-_{\{\a_1,\dots,\a_R\}\cup\{\b_1,\dots,\b_S\}}\right].
\end{equation}
We shall reformulate this expression by means of a few identities, involving the quantities \eqref{DefEpm} and \eqref{DefApm} and their relations to the Izergin and Slavnov determinants, and  that we now derive.

\bigskip

We start with some preliminary identity which relates the two functionals $\mathcal{A}_{\{ x\}}^+[f]$ and $\mathcal{A}_{\{ x\}}^-[f']$ \eqref{DefApm} when $f'$ is related to $f$ via a ratio of the two functions $E_{\{ x\} }^+$ and $E_{\{ x\} }^-$ \eqref{DefEpm}:
\begin{identity}\label{id-1}
For any set $\{x\}\equiv\{x_1,\dots ,x_M\}$ of arbitrary complex numbers, we have
\begin{equation} \label{Identity-q}
   \mathcal{A}_{\{ x\}}^{\pm }[ f ] 
   =\mathcal{A}_{\{ x\}}^{\mp }\!
    \left[ -\frac{ E_{ \{ x\} }^{\pm }}{E_{\left\{x\right\} }^{\mp }}f\right].
\end{equation}
\end{identity}

\begin{proof}
 To prove this identity we use an expression for the determinant of the sum of two matrices, 
\begin{align}
V(\{x\})\ &\mathcal{A}_{\{ x\} }^{-}\!\left[ \frac{E_{\{ x\} }^{+}}{E_{\{x\} }^{-}}f\right]\nonumber\\
&=\sum_{h_{1}=0}^{1}\cdots\sum_{h_{M}=0}^{1}\prod_{n=1}^{M}
\left( - f(x_{n})\prod_{m=1}^{M}\frac{x_{n}-x_{m}+\eta}{x_{n}-x_{m}-\eta}\right)^{\! h_{n}}\, V(\{x-h\eta\})\nonumber\\
&=\sum_{h_{1}=0}^{1}\cdots \sum_{h_{M}=0}^{1}\left( \prod_{n=1}^{M}
 f(x_{n})\right)\, V(\{x+h\eta\})\nonumber\\
&=V(\{x\})\ \mathcal{A}_{\{ x\} }^{+}[ f]\vphantom{\prod_{n=1}^{M}}.
\end{align}
We have also used here the identity \eqref{Vandermonde_trick} for the Vandermonde determinants. 
\end{proof}

\bigskip

We shall now formulate a generalization of a result of Kostov  \cite{Kos120,Kos12,KosM12,BetK14,Kos12} concerning the relation between the functionals $\mathcal{A}_{\{x\}}^\pm[ f ]$ and the Izergin determinant (see also \cite{Garb14}).
 
For $\mu\in\mathbb{C}$ and two sets  $\{x\}\equiv\{x_1,\dots,x_N\}$ and $\{y\}\equiv\{y_1,\dots,y_N\}$ of arbitrary complex numbers, we introduce the following function, that we call {\em generalized Izergin determinant}:
 \begin{equation}\label{Ize-det}
 \mathcal{I}^{(\mu)}_N(\{x\},\{y\})=\frac{\pl_{a,b=1}^N(x_a-y_b+\eta)}{V(x_1,\dots, x_N)\, V(y_N,\dots,y_1)}\, \det_N \big[ t_\mu(x_a-y_b)\big],
 \end{equation}
 where we have defined
 \begin{equation}
 \label{t_mu}
 t_\mu(x)=\frac\mu x-\frac 1{x+\eta}.
 \end{equation}
 For $\mu=1$ this formula gives the partition function of the rational six-vertex model with domain wall boundary conditions \cite{Ize87}.
 
\begin{identity}\label{prop-Izergin-rel}
 For $\mu\in\mathbb{C}$ and any sets $\{x_1,\dots,x_N\}$ and $\{y_1,\dots,y_N\}$ of arbitrary complex numbers, we have
\begin{align}
  \mathcal{I}^{(\mu)}_N(\{x\},\{y\})
  &=(-1)^N \mathcal{A}^-_{\{x\}}\!\left[\mu E^+_{\{y\}}\right] \label{Iz-rel1}\\
  &=(-1)^N \mathcal{A}^+_{\{y\}}\!\left[\mu E^-_{\{x\}}\right].\label{Izergin_relation}
\end{align}
\end{identity}

\begin{proof}
 We shall prove here the first equality, the second one can be proven  in the same way (or by replacing $\eta $ by $-\eta$ in the Izergin determinant).
 
Let us first introduce the following auxiliary matrix with elements,
\begin{equation}
  \mathcal{M}_{ab}= \pl_{k=1}^N(x_a-y_k+\eta)\, t_\mu(x_a-y_b),
\end{equation}
so as to rewrite \eqref{Ize-det} as
\begin{equation}
   \mathcal{I}^{(\mu)}_N(\{x\},\{y\})=\frac{\det_N \mathcal{M}}{V(x_1,\dots, x_N)\, V(y_N,\dots,y_1)}.
\end{equation}
We also introduce a set of $N$ polynomials of degree $N-1$,
\begin{equation}
     Z_b(x)=\pl_{k\neq b} (x-y_k+\eta)=\sul_{j=1}^N \mathcal{C}_{j,b} \,x^{j-1},
     \qquad 1\le b\le N.
\end{equation}
The coefficients $\mathcal{C}_{j,b}$ of these polynomials can be seen as the entries of an $N\times N$ matrix $\mathcal{C}$. It is easy to observe that
\begin{align}
\mathcal{M}_{ab}
   &=\mu  Z_b(x_a-\eta)\pl_{k=1}^N\frac {x_a-y_k+\eta}{x_a-y_k}-Z_b(x_a)\nonumber\\
   &=\sul_{j=1}^N\left(\mu E^+_{\{y\}}(x_a)\,  (x_a-\eta)^{j-1}-x_a^{j-1}\right)\mathcal{C}_{j,b},
\end{align}
which means that $\mathcal{M}$ can be factorized into a product of two matrices, so that
\begin{equation}
   \det_N \mathcal{M}
   =(-1)^N\, V(\{x\})\ \mathcal{A}^-_{\{x\}}\!\left[\mu E^+_{\{y\}}\right]\,\det_N \mathcal{C}.
 \end{equation}
It now remains to compute the determinant of the matrix $\mathcal{C}$. It is easy to see from the definition of the coefficients $\mathcal{C}_{j,b}$ that the product  of $\mathcal{C}$ with a Vandermonde matrix is diagonal:
\begin{equation}
\sul_{j=1}^N (y_a-\eta)^{j-1}\, \mathcal{C}_{j,b}=\delta_{a,b}\pl_{j\neq a}(y_a-y_j),
\end{equation}
which means that
$\det_N \mathcal{C}= V(y_N,\dots,y_1)$.
\end{proof} 

Note that the Izergin determinant \eqref{Ize-det} is defined only in the case where the cardinality of the two sets of parameters $\{x\}$ and $\{y\}$ are equal, so that Identity~\ref{prop-Izergin-rel} is a priori valid only in this case. To overcome this restriction, it is however possible to generalize the part of Identity~\ref{prop-Izergin-rel} which concerns the relation between the two functions $ \mathcal{A}^-_{\{x\}} \big[\mu E^+_{\{y\}}\big]$ and $\mathcal{A}^+_{\{y\}}\big[\mu E^-_{\{x\}}\big]$ to cases in which the number of parameters $x$ and $y$ does not coincide.

 \begin{identitybis}{prop-Izergin-rel}
 \label{limit_trick}
For two sets $\{x_1,\dots,x_M\}$ and $\{y_1,\dots, y_N\}$ of arbitrary complex numbers and $\mu\neq 1$, we have
 \begin{equation}\label{id-Mnot=N}
 \mathcal{A}^+_{\{y\}}\!\left[\mu E^-_{\{x\}}\right]
 =(1-\mu)^{N-M}\mathcal{A}^-_{\{x\}}\!\left[\mu E^+_{\{y\}}\right] .
 \end{equation}
 \end{identitybis}
 
\begin{proof}
 The case $N=M$ follows immediately from Identity~\ref{prop-Izergin-rel}.
 Let us therefore suppose that $M>N$ (the case $M<N$ can be proven in the same way starting from the right hand side). The proof is based on the following simple observation\footnote{Here and in the following we mean that the limits are taken independently one at the time.},
\begin{equation}
\mathcal{A}^+_{\{y\}}\!\left[\mu E^-_{\{x\}}\right]
=(1-\mu)^{N-M}\lim_{v_1,\dots,v_{M-N}\rightarrow\infty} 
\mathcal{A}^+_{\{y\}\cup\{v_1,\dots,v_{M-N}\}}\!\left[\mu E^-_{\{x\}}\right],
\end{equation}
from which we can use  Identity~\ref{prop-Izergin-rel} and obtain
\begin{equation}
\mathcal{A}^+_{\{y\}}\left[\mu E^-_{\{x\}}\right]
=(1-\mu)^{N-M}\lim_{v_1,\dots,v_{M-N}\rightarrow\infty} 
\mathcal{A}^-_{\{x\}}\left[\mu E^+_{\{y\}\cup\{v_1,\dots,v_{M-N}\}}\right].
\end{equation}
 Computing the limits we obtain \eqref{id-Mnot=N}.
 \end{proof}
 
Note that for $N>M$ both sides of \eqref{id-Mnot=N} are polynomials in $\mu$, so that the identity holds also in the case $\mu=1$. This leads to the following simple observation: 

\begin{corollary}
  \label{zero_overlap}
For arbitrary sets  $\{x_1,\dots,x_M\}$ and $\{y_1,\dots, y_N\}$ of complex numbers such that $N>M$, we have
\begin{equation}
\mathcal{A}^\pm_{\{y\}}\!\left[ E^\mp_{\{x\}}\right]=0.
\end{equation}
\end{corollary}

\bigskip
 
The next identity that we shall derive will enable us to relate the scalar product of two separate states, one of them being an eigenstate of the transfer matrix,  with the Slavnov formula 
 \cite{Sla89} for the scalar product of an off-shell and an on-shell Bethe vectors in the framework of the algebraic Bethe ansatz.
 
For $\mu\in\mathbb{C}$, for two sets of parameters $\{x\}\equiv\{x_1,\dots, x_M\}$ and $\{y\}\equiv\{y_1,\dots, y_M\}$, and a set of inhomogeneity parameters $\{\xi\}\equiv\{\xi_1,\ldots,\xi_N\}$ (which may possibly coincide) 
with $M\le N$, we introduce the following function, that we call {\em Slavnov determinant}:
\begin{equation}\label{Slavnov-det}
 \mathcal{S}_M^{(\mu)}(\{x\},\{y\}|\{\xi\})
 =\frac{\pl_{j,k=1}^M(x_j-y_k+\eta)}{V(x_1,\dots, x_M)\, V(y_M,\dots,y_1)}\,
   \det_M\mathcal{H}^{(\mu)}(\{x\},\{y\}|\{\xi\}).
\end{equation}
In \eqref{Slavnov-det}, the matrix $\mathcal{H}^{(\mu)}(\{x\},\{y\}|\{\xi\})$ is defined by its elements as
\begin{equation}
   \big[\mathcal{H}^{(\mu)}(\{x\},\{y\}|\{\xi\})\big]_{j k}
   =\mu E^+_{\{\xi\}}(y_k)\,t(x_j-y_k)-  \frac{E^+_{\{x\}}(y_k)}{E^-_{\{x\}}(y_k)}\,t(y_k-x_j),
\end{equation}
where $t(x)\equiv t_1(x)$ is given by \eqref{t_mu}.

As already mentioned, this formula appeared initially  in the ABA framework as a representation for the scalar product of a Bethe state associated with generic parameters $\{y\}$ (off-shell Bethe state) and a Bethe state associated with Bethe roots $\{x\}$  satisfying a system of Bethe equations (on-shell Bethe state) for an inhomogeneous model with inhomogeneities $\xi_1,\dots,\xi_N$ and with a twist $\mu$:
\begin{equation}
 \label{Bethe_equation}
 \mu E^+_{\{\xi\}}(x_m)=-\frac{E^+_{\{x\}}(x_m)}{E^-_{\{x\}}(x_m)}.
 \end{equation}
Note at this point that the Bethe equations \eqref{Bethe_anti} introduced in the previous section correspond to a special case of \eqref{Bethe_equation} with $\mu=-1$.

\begin{identity}
  \label{Slavnov_relation_th}
For $\mu\in\mathbb{C}$, for $\{x_1,\dots, x_M\}$ a solution to the Bethe equations \eqref{Bethe_equation} associated with a set of inhomogeneity parameters $\{\xi_1,\dots,\xi_N\}$,  and for $\{y_1,\dots, y_M\}$ a set of arbitrary complex numbers, we have
\begin{equation}
 \label{Slavnov_relation}
 \mathcal{S}_M^{(\mu)}(\{x\},\{y\}|\{\xi\})=\mathcal{A}^-_{\{x\}\cup\{y\}}\left[\mu E^+_{\{\xi\}}\right].
\end{equation}
\end{identity}

\begin{proof}
Let us first suppose that the parameters $x_j$ are pairwise distinct, and let $Q(\la)$ be the normalized polynomial of degree $M$ with roots $x_1,\ldots,x_M$:
\begin{equation}
 Q(\la)=\pl_{j=1}^M(\la-x_j).
\end{equation}
We  define from $Q(\lambda)$ $M$ different polynomials $Q_k(\la)$ of degree $M-1$ as follows:
\begin{equation}
Q_k(\la)=\pl_{j\neq k}(\la-x_j), \qquad 1\le k \le M.
\end{equation}
Then, similarly as what has been done in the proof of Identity~\ref{prop-Izergin-rel}, we can define a $2M\times 2M$ matrix $\mathcal{C}$ from the coefficients of the following $2M$ polynomials:
\begin{alignat}{2}
&Q_k(\la)\, Q_k(\la+\eta)=\sul_{a=1}^{2M}\mathcal{C}_{k,a}\,\la^{a-1}, &
\qquad &1\le k \le M,\\
&Q_k(\la)\, Q(\la+\eta)=\sul_{a=1}^{2M}\mathcal{C}_{k+M,a}\,\la^{a-1}, &
\qquad &1\le k \le M.
\end{alignat}
Evidently the coefficients $\mathcal{C}_{k,2M}$ all vanish for $k=1,\dots,M$, but the $2M\times 2M$ matrix $\mathcal{C}$ is invertible, and it is not difficult to compute its determinant using the following identities:
\begin{gather}
  \sul_{a=1}^{2M}\mathcal{C}_{b,a}\, (x_k-\eta)^{a-1}
  =\delta_{b,k}\,  Q_k (x_k)\, Q_k (x_k-\eta),
  \\
  \sul_{a=1}^{2M}\mathcal{C}_{b,a}\, x_k^{a-1}
  = \delta_{b,k+M} \, Q(x_k+\eta)\, Q_k (x_k)+\delta_{b,k} \, Q_k (x_k) \, Q_k (x_k+\eta).
\end{gather}
It means that the product of the matrix $\mathcal{C}$ with the Vandermonde matrix constructed from the variables $x_1-\eta,\ldots,x_M-\eta,x_1,\ldots,x_M$ is triangular, and therefore the determinant of the matrix $\mathcal{C}$ is given as
\begin{align}
\det_{2M}\mathcal{C}
&= \frac{\pl_{k=1}^M \left[ Q^2_k (x_k)\, Q_k (x_k-\eta)\, Q (x_k+\eta) \right] }{V(x_1-\eta,\dots,x_M-\eta,x_1\dots,x_M)}
\nonumber\\
&=(-1)^{\frac{M(M+1)}{2}}\pl_{k=1}^M \big[ Q_k (x_k)\, Q_k (x_k+\eta)\big].
\end{align}

On the other hand let us compute the following product,  
\begin{equation}\label{prodA-C}
    \mathcal{A}^-_{\{x\}\cup\{y\}}\!\left[\mu E^+_{\{\xi\}}\right]\ \det_{2M} \mathcal{C}
    =\frac{\det_{2M}\left(
 \begin{array}{cc}
 \mathcal{G}^{(1,1)}&\mathcal{G}^{(1,2)}\\\mathcal{G}^{(2,1)}&\mathcal{G}^{(2,2)}\end{array}\right)}{V(\{x\}\cup\{y\})},
\end{equation}
where the matrix in the numerator is presented in a block form, $\mathcal{G}^{(p,q)}$ ($1\le p,q\le 2$) being  $M\times M$ matrices which can easily be computed. For instance, the upper left block $\mathcal{G}^{(1,1)}$ has for elements
\begin{align}
  \big[\mathcal{G}^{(1,1)}\big]_{j,k}
  &=\sul_{a=1}^{2M}\mathcal{C}_{j,a}\left(x_k^{a-1}-\mu  E^+_{\{\xi\}}(x_k)\, (x_k-\eta)^{a-1}\right)\nonumber\\
  &= Q_j(x_k)\, Q_j(x_k+\eta)-\mu E^+_{\{\xi\}}(x_k)\, Q_j(x_k)\, Q_j(x_k-\eta)
  \nonumber\\
  &= \delta_{jk}\, Q_k(x_k) \left(Q_k(x_k+\eta)-\mu E^+_{\{\xi\}}(x_k)\, Q_k(x_k-\eta)\right) ,
\end{align}
so that $ \mathcal{G}^{(1,1)}=0$ due to the Bethe equations \eqref{Bethe_equation}. It means in particular that we do not need to compute the lower right block $ \mathcal{G}^{(2,2)}$ to have access to the quantity \eqref{prodA-C}. The lower left block  $\mathcal{G}^{(2,1)}$ is a diagonal matrix,
\begin{align}
  \big[ \mathcal{G}^{(2,1)} \big]_{j,k}
   &=\sul_{a=1}^{2M}\mathcal{C}_{j+M,a}\left(x_k^{a-1}-\mu  E^+_{\{\xi\}}(x_k)\, (x_k-\eta)^{a-1}\right)\nonumber\\
   &=Q_j(x_k)\, Q(x_k+\eta)-\mu E^+_{\{\xi\}}(x_k)\, Q(x_k)\, Q_j(x_k-\eta)\nonumber\\
   &=\delta_{jk}\, Q_k(x_k)\,Q(x_k+\eta).
\end{align}
Finally, the upper right block $\mathcal{G}^{(1,2)}$ is given by
\begin{align}
     \big[\mathcal{G}^{(1,2)}\big]_{j,k}
     &=\sul_{a=1}^{2M}\mathcal{C}_{j,a}\left(y_k^{a-1}-\mu  E^+_{\{\xi\}}(y_k) \, (y_k-\eta)^{a-1}\right)
     \nonumber\\ 
    &=Q_j(y_k)\, Q_j(y_k+\eta)-\mu E^+_{\{\xi\}}(y_k)\, Q_j(y_k) \, Q_j(y_k-\eta)\nonumber\\
    &=-\frac 1\eta\, Q(y_k)\, Q(y_k-\eta) \,  \big[\mathcal{H}^{(\mu)}(\{x\},\{y\}|\{\xi\})\big]_{j k}.
\end{align}
It means that
\begin{multline}
   \mathcal{A}^-_{\{x\}\cup\{y\}}\!\left[\mu E^+_{\{\xi\}}\right]\, \det_{2M} \mathcal{C}
   =\frac{\pl_{k=1}^M \left[ Q(y_k)\, Q(y_k-\eta)\, Q_k(x_k)\,Q_k(x_k+\eta)\right]}{V(\{x\}\cup\{y\})}\,
   \\
   \times \det_M\mathcal{H}(\{x\},\{y\}|\{\xi\}),
\end{multline}
which leads directly to the identity \eqref{Slavnov_relation}.

Note that both $\mathcal{S}_M^{(\mu)}(\{x\},\{y\}|\{\xi\})$ and $\mathcal{A}^-_{\{x\}\cup\{y\}}\left[\mu E^+_{\{\xi\}}\right]$ formally contain zero over zero terms if computed for two or more coinciding $x_i$. So that the identity \eqref{Slavnov-det} has to be meant as a limit to the Bethe roots in the case of two or several coinciding roots. In fact we can introduce a set of parameters $\{x^{(\epsilon)}\}$ pairwise distinct for $\epsilon>0$ and converging to the solution of the Bethe equations $\{x\}$ when $\epsilon\rightarrow 0$. We have to prove that the l.h.s. and r.h.s. of \eqref{Slavnov-det} are finite and coincide under this limit. The identity \ref{id-Mnot=N}, being proven for arbitrary values of the parameters, implies that $\mathcal{A}^-_{\{x^{(\epsilon)}\}\cup\{y\}}\left[\mu E^+_{\{\xi\}}\right]$ has a smooth limit for $\epsilon\rightarrow 0$. Now we can repeat the proof above developed for pairwise distinct roots checking that at any step all remains finite and that one just reproduces $\mathcal{S}_M^{(\mu)}(\{x^{(\epsilon)}\},\{y\}|\{\xi\})$ for $\epsilon\rightarrow 0$. 
 \end{proof} 
 
 If the number $M$ of parameters $x$ is exactly half of the length of the chain we can directly apply Identity~\ref{prop-Izergin-rel} and Identity~\ref{Slavnov_relation_th} to obtain the following result:
 
 \begin{corollary}
 Let $N=2M$. Then, for $\mu\in\mathbb{C}$, for $\{x_1,\dots, x_M\}$ a solution to the Bethe equations \eqref{Bethe_equation} associated with a set of inhomogeneity parameters $\{\xi_1,\dots,\xi_N\}$,  and for $\{y_1,\dots, y_M\}$ a set of arbitrary complex numbers, we have
  \begin{align}
 \label{Slavnov_relation2}
 \mathcal{S}_M^{(\mu)}(\{x\},\{y\}|\{\xi\})
 &=\mathcal{A}^+_{\{\xi\}}\left[\mu E^-_{\{x\}\cup\{y\}}\right]
 \\
 &=\mathcal{I}^{(\mu)}_N(\{x\}\cup\{y\},\{\xi\}).
 \end{align}
 \end{corollary}
 
 This relation is of prime importance for the computation of the scalar products of separate states. The inconvenient restriction  on the number of parameters $x$ and $y$ can be relaxed thanks to Identity~\ref{limit_trick}, which enables us to formulate the following generalization of \eqref{Slavnov_relation2}:
\begin{corollary}
   Let $\{x_1,\dots, x_M\}$ be a  solution of the Bethe equations \eqref{Bethe_equation} associated with the set of inhomogeneity parameters $\{\xi_1,\ldots,\xi_N\}$ and the twist $\mu$,  and let $\{y_1,\dots, y_M\}$ be a set of arbitrary complex numbers. Then the following relation holds:
  \begin{equation}
   \mathcal{S}_M^{(\mu)}(\{x\},\{y\}|\{\xi\})=(1-\mu)^{2M-N}\mathcal{A}^+_{\{\xi\}}\left[\mu E^-_{\{x\}\cup\{y\}}\right].
  \end{equation}
\end{corollary}

\bigskip

Finally, we would like to generalize the previous results to cases for which the number of parameters $x$ and $y$ are not obligatory the same. To this aim we introduce the following generalization of the Slavnov determinant \eqref{Slavnov-det}, which is defined for two sets of parameters $\{x\}\equiv\{x_1,\dots,x_M\}$ and $\{y\}\equiv\{y_1,\dots, y_{M+S}\}$, with $S\ge 0$,  a set of inhomogeneity parameters $\{\xi\}\equiv\{\xi_1,\ldots,\xi_N\}$ and a twist $\mu$: 
 \begin{multline}\label{Slavnov-gen}
 \mathcal{S}_{M,M+S}^{(\mu)}(\{x\},\{y\}|\{\xi\})
 =\frac{\pl_{j=1}^M\pl_{k=1}^{M+S}(x_j-y_k+\eta)}{V(x_1,\dots, x_M)\, V(y_{M+S},\dots,y_1)}\\
 \times
 \det_{M+S}\widetilde{\mathcal{H}}^{(\mu)}(\{x\},\{y\}|\{\xi\}),
 \end{multline}
 with
 \begin{alignat*}{2}
 \big[ \widetilde{\mathcal{H}}^{(\mu)}(\{x\},\{y\}|\{\xi\})\big]_{j k}
 &=\mu E^+_{\{\xi\}}\!(y_k)\,t(x_j-y_k)-  \frac{E^+_{\{x\}}\!(y_k)}{E^-_{\{x\}}\!(y_k)}\, t(y_k-x_j),
       &\  &\text{if}\  j\le M,
       \nonumber\\
 &=\mu E^+_{\{\xi\}}\!(y_k)\,y_k^{j-M-1}-  \frac{E^+_{\{x\}}\!(y_k)}{E^-_{\{x\}}\!(y_k)} (y_k+\eta)^{j-M-1},
       &\ \,&\text{if}\  j> M.
 \end{alignat*}
This type of object was first introduced in \cite{FodW12}. This enables us to formulate the last identity of this subsection:

 \begin{identity}
   \label{gen_Slavnov_relation_th}
Let $\{x_1,\dots, x_M\}$ be a  solution of the Bethe equations \eqref{Bethe_equation} associated with the set of inhomogeneity parameters $\{\xi_1,\ldots,\xi_N\}$ and the twist $\mu$,  and let $\{y_1,\dots, y_{M+S}\}$ be a set of arbitrary complex numbers. Then,
\begin{equation}
 \label{gen_Slavnov_relation}
 \mathcal{S}_{M,M+S}^{(\mu)}(\{x\},\{y\}|\{\xi\})=\mathcal{A}^-_{\{x\}\cup\{y\}}\!\left[\mu E^+_{\{\xi\}}\right]. \end{equation}
\end{identity}
 
The proof of this identity follows the same lines as the proof of \eqref{Slavnov_relation}. It is however more cumbersome, so that we  give its details in Appendix A.

\subsection{An alternative representation for the scalar product of two separate states}

We shall now use the identities that have been derived in the previous subsection to rewrite the representation \eqref{sp-sep} (or \eqref{sp-rep2}) for the scalar product of two generic separate states $\bra{\alpha}$ and $\ket{\beta}$ into a form for which the consideration of the homogeneous limit is completely straightforward.  

Let us first notice that Identity~\ref{id-1} enables us to rewrite \eqref{sp-rep2} into a slightly simpler form:

\begin{proposition}\label{prop-sp-sep}
Let $\a(\la)$ and $\b(\la)$ be two polynomials of respective degree $R$ and $S$, and which can be factorized as
\begin{equation}
\alpha (\lambda )=\prod_{n=1}^{R}(\lambda -\alpha _{n}),\qquad
\beta(\lambda )=\prod_{m=1}^{S}(\lambda -\beta _{m}),
\end{equation}
in terms of some sets of roots $\{\alpha_1,\ldots,\alpha_R\}$ and $\{\beta_1,\ldots,\beta_S\}$.
Then the scalar product of the corresponding separate states $\bra{\alpha}$ and $\ket{\beta}$ can be written in the following form:
\begin{equation}
  \moy{ \alpha\, |\, \beta }
   =(-1)^{N(R+S)}\prod_{j=1}^R d(\alpha_j)\prod_{k=1}^S d(\beta_k) \ 
   \mathcal{A}^+_{\{\xi\}}\!\left[ -E^-_{\{\a_1,\dots,\a_R\}\cup\{\b_1,\dots,\b_S\}}\right],
\label{scalar_product_A}
\end{equation}
where we have used the notations \eqref{DefEpm}-\eqref{DefApm}.
\end{proposition}

We can now use the identities~\ref{prop-Izergin-rel} and~\ref{limit_trick}, and we obtain the following result: 

\begin{theorem}
Under the same hypothesis and notations as in Proposition~\ref{prop-sp-sep}, the scalar product of the two separate states $\bra{\alpha}$ and $\ket{\beta}$ can be written as:
\begin{multline}
  \moy{ \alpha\, |\, \beta }
   =(-1)^{N(R+S)}\,2^{N-(R+S)}\prod_{j=1}^R d(\alpha_j)\prod_{k=1}^S d(\beta_k)\\
   \times 
   \mathcal{A}^-_{\{\a_1,\dots,\a_R\}\cup\{\b_1,\dots,\b_S\}}\!\left[ -E^-_{\{\xi\}}\right].
\label{scalar_product_B}
\end{multline}
In the particular case where $R+S=N$, the scalar product of the two separate states $\bra{\alpha}$ and $\ket{\beta}$ can be written in terms of the generalized Izergin determinant \eqref{Ize-det} as:
\begin{multline}
  \moy{ \alpha\, |\, \beta }
   =(-1)^{N(R+S+1)}\prod_{j=1}^R d(\alpha_j)\prod_{k=1}^S d(\beta_k)\\
   \times \mathcal{I}_N^{(-1)}(\{\alpha_1,\ldots,\alpha_R\}\cup\{\beta_1,\ldots,\beta_S\},\{\xi\}). 
\label{scalar_product_C}
\end{multline}
\end{theorem}

Note that, as announced before, the representation \eqref{scalar_product_B} is completely smooth with respect to the homogeneous limit in which all the inhomogeneity parameters $\xi_j$, $1\le j\le N$, tend to the same value.

\subsection{The scalar product of  a generic separate state with an eigenstate of the transfer matrix} 

We shall now focus on the case in which one of the two separate states is an eigenstate $\ket{Q_\tau}$ of the transfer matrix. As we shall see, the knowledge of such scalars products is indeed sufficient to obtain some adequate representations for the form factors of local operators. This fact has to be put in relation with what happens in the ABA framework, where the consideration of the scalar products of off-shell and on-shell Bethe vectors is used to compute the local spin form factors \cite{KitMT99}. 

In the whole subsection, $\bra{\alpha}$ will denote a separate state associated with a given polynomial $\alpha(\lambda)$ of degree $M$,
\begin{equation}
\alpha (\lambda )=\prod_{n=1}^{M}(\lambda -\alpha _{n}),
\end{equation}
and $\ket{Q_\tau}$ will denote a given eigenstate of the antiperiodic transfer matrix, associated with an eigenvalue $\tau(\lambda)$. We recall that such an eigenstate can be written in two different forms, either by using the polynomial
\begin{equation}
   Q_\tau(\lambda)=\pl_{k=1}^R(\la-\la_k)
\end{equation}
satisfying the functional $T$-$Q$ equation \eqref{hom-Baxter-Eq} with $\tau(\lambda)$ itself, or by using the polynomial
\begin{equation}
   Q_{-\tau}(\lambda)=\pl_{k=1}^{N-R}(\la-\widehat{\la}_k)
\end{equation}
satisfying the $T$-$Q$-equation with $-\tau(\lambda)$ (see Proposition~\ref{prop-eigen-ABA}).
We also recall that in this framework the set of roots $\{\la_1,\dots,\la_R \}$ of $Q_\tau(\lambda)$ satisfies the Bethe equations \eqref{Bethe_equation} with $\mu=-1$ (as well as the set of roots $\{\widehat{\la}_1,\dots,\widehat{\la}_{N-R} \}$ of $Q_{-\tau}(\lambda)$). 

\begin{proposition}
The scalar product $\moy{ \a\,|\, Q_\tau}$ can be written in the two possible following forms:
\begin{align}
   \moy{ \a\,|\, Q_\tau} 
   &=(-1)^{N(R+M)}\, \pl_{n=1}^M d(\a_n)\, \pl_{k=1}^R d(\la_k) \ 
        \mathcal{A}^+_{\{\xi\}}\!\left[ -E^-_{\{\la\}\cup\{\a\}}\right]
         \label{withlambda}\\
   &=(-1)^{N(N-R+M)}\, \pl_{n=1}^M d(\a_n)\, \pl_{k=1}^{N-R} d(\widehat{\la}_k) \ 
         \mathcal{A}^+_{\{\xi\}}\!\left[ E^-_{\{\widehat{\la}\}\cup\{\a\}}\right] .
         \label{withlambdahat}
\end{align}
\end{proposition}

\begin{proof}
The first line is a direct consequence of \eqref{scalar_product_A}. The second one can be obtain in the same way from the fact that
\begin{equation}
   \frac{Q_\tau(\xi_a-\eta)}{Q_\tau(\xi_a)}=-\frac{Q_{-\tau}(\xi_a-\eta)}{Q_{-\tau}(\xi_a)}.
\end{equation}
\end{proof}

We can now use the identities obtained in the previous subsection to prove the following results for the scalar product of the separate state $\bra{\alpha}$ with the eigenstate $\ket{Q_\tau}$.

\begin{theorem}\label{th-sc-pdt}
If $M<R$ the scalar product between $\bra{\alpha}$ and  $\ket{Q_\tau}$ vanishes:
\begin{equation}
   \moy{ \a\,|\, Q_\tau} =0.
\label{sp_less}
\end{equation}
If $M=R$ the scalar product  between $\bra{\alpha}$ and  $\ket{Q_\tau}$  can be written either as an Izergin determinant \eqref{Ize-det},
\begin{equation}
   \moy{ \a\,|\, Q_\tau} =\left(\pl_{n=1}^M d(\a_n)\pl_{k=1}^{N-M} d(\widehat{\la}_k)\right) \ 
 \mathcal{I}^{(1)}_N(\{\a\}\cup\{\widehat{\la}\},\{\xi\}),
 \label{sp_izergin}
\end{equation}
or as a Slavnov determinant \eqref{Slavnov-det},
\begin{equation}
   \moy{ \a\,|\, Q_\tau} =(-1)^M\, 2^{N-2M}\left(\pl_{n=1}^M d(\a_n)\, d(\la_n)\right) \,
 \mathcal{S}_M^{(-1)}(\{\la\},\{\a\}|\{\xi\}).
 \label{sp_slavnov}
\end{equation}
If $M>R$ the scalar product  between $\bra{\alpha}$ and  $\ket{Q_\tau}$ can be written as a generalized Slavnov determinant \eqref{Slavnov-gen}:
\begin{equation}
  \moy{ \a\,|\, Q_\tau}=(-1)^R\, 2^{N-M-R}\left(\pl_{n=1}^M d(\a_n)\pl_{k=1}^R d(\la_k)\right) \,
 \mathcal{S}_{R,M}^{(-1)}(\{\la\},\{\a\}|\{\xi\}).
 \label{sp_more}
\end{equation}
\end{theorem}

\begin{proof}
The proof of this theorem is straightforward. To prove  (\ref{sp_less}) it is enough to use (\ref{withlambdahat}) and Corollary \ref{zero_overlap}; (\ref{sp_izergin}) follows from (\ref{withlambdahat}) and Identity~\ref{prop-Izergin-rel}. The two last representations follow directly from (\ref{withlambda}), and from the identities~\ref{limit_trick}, \ref{Slavnov_relation_th} and  \ref{gen_Slavnov_relation_th}.
\end{proof}

It is important to mention that all the representations of Theorem~\ref{th-sc-pdt} for the scalar product of the eigenstate $\ket{Q_\tau}$ with the generic separate state $\bra{\alpha}$ remain finite (and manageable) in the homogeneous limit.

Note finally to conclude this section that the fact that we can represent the scalar product $\moy{ \a\,|\, Q_\tau}$ as a Slavnov determinant allows us to take, as usual, the limit in which the two states are equal. We therefore obtain, as in ABA, a representation for the ``square of the norm''of the eigenstate $\ket{Q_\tau}$ (i.e. for the scalar product $\moy{Q_\tau\,|\, Q_\tau}$) in terms of a Gaudin determinant.

\begin{corollary}\label{cor-norm}
The ``square of the norm'' of the eigenstate $\ket{Q_\tau}$ of the antiperiodic transfer matrix is given by
\begin{equation}
\label{gaudin}
    \moy{Q_\tau\,|\, Q_\tau}=  2^{N-2R}\left(\pl_{n=1}^R d(\la_n)\right)^{\!2}\,
    \frac{\pl_{m,n=1}^R(\la_m-\la_n+\eta)}{\pl_{m\neq n}(\la_m-\la_n)}\ \det_R \Phi_\tau
\end{equation}
with 
\begin{equation}
\big[\Phi_\tau\big]_{m,n}
=\frac{\partial}{\partial\la_n}\log\left( \frac {a(\la_m)}{d(\la_m)}\pl_{b=1}^R\frac{\la_m-\la_b-\eta}{\la_m-\la_b+\eta}\right), \quad 1\le m,n \le R.
\end{equation}
\end{corollary}

\section {Form factors of local operators}

In this section we compute the form factors, i.e. the matrix elements of the local spin operators between the eigenstates of the transfer matrix. 

Let us first consider the local operator $\s_n^-$.
Its matrix elements between two eigenstates $\bra{Q_\tau}$ and $\ket{Q_{\tau'}}$ of the transfer matrix can be computed by acting with this operator on one of these two states, which can be done as usual by means of the solution of the quantum inverse problem \cite{KitMT99,MaiT00}. The reconstruction formula \cite{KitMT99,Nic13}  for this operator takes the form,
\begin{equation}
\s_n^-= (-1)^N \left(\pl_{j=1}^{n-1}\frac{\mathcal{T}(\xi_j)}{a(\xi_j)}\right)\, \frac{D(\xi_n)}{a(\xi_n)}\left(\pl_{j=n+1}^N \frac{\mathcal{T}(\xi_j)}{a(\xi_j)}\right).
\label{inverse_problem}
\end{equation}
It enables us to formulate the following result:

\begin{theorem}\label{th-ffz}
Let $\ket{Q_\tau}$ and $\ket{Q_{\tau'}}$ be two  eigenstates of the transfer matrix, with 
\begin{equation}\label{pol-tau-tau'}
  Q_\tau(\la)=\pl_{k=1}^R(\la-\la_k),\qquad Q_{\tau'}(\la)=\pl_{k=1}^{R'}(\la-\la'_k).
\end{equation}
Then the corresponding form factors for the operator $\sigma_n^-$ are given by the following expressions:
\begin{itemize} 
\item if $|R-R'|> 1$, the form factor $\bra{Q_\tau}\,\s_n^-\,\ket{Q_{\tau'}}$ vanishes,
\begin{equation}
\bra{Q_\tau}\,\s_n^-\,\ket{Q_{\tau'}}=0\, ;
\end{equation}
\item if $R=R'+ 1$, the form factor $\bra{Q_\tau}\,\s_n^-\,\ket{Q_{\tau'}}$ can be expressed in terms of the determinant of a matrix of size $R$ as
\begin{multline}
    \bra{Q_\tau}\,\s_n^-\,\ket{Q_{\tau'}}= 2^{N-2R}\, (-1)^{N-1}\,
    \frac{Q_\tau(\xi_n)}{Q_{\tau'}(\xi_n)}
    \\
    \times
    \frac{\pl_{k=1}^R a_n(\la_k) \pl_{j=1}^{R'}d_n(\la'_j)}{V(\la_1,\dots,\la_R)\, V(\la'_{R'},\dots,\la'_1)}\ 
    \det_R\mathcal{F}^{-}(\{\la\},\{\la'\},\xi_n),\vphantom{ \pl_{j=1}^{R'}}
\label{ff_+1}
\end{multline}
where
\begin{align}
  &a_n(\la)=\pl_{j=1}^{n}(\la-\xi_j+\eta)\pl_{j=n+1}^{N}(\la-\xi_j),\\
  &d_n(\la)=\pl_{j=1}^{n}(\la-\xi_j)\pl_{j=n+1}^{N}(\la-\xi_j+\eta),
\end{align}
and
\begin{align}
\mathcal{F}^{-}_{j,k} 
   &=\frac{a(\la'_k)}{d(\la'_k)}\, Q_\tau(\la'_k-\eta)\, t(\la_j-\la'_k)+Q_\tau(\la'_k+\eta)\, t(\la'_k-\la_j), 
   \quad\text{for}\  k<R,
   \nonumber \\
\mathcal{F}^{-}_{j,R}
    &=t(\la_j-\xi_n)\vphantom{ \pl_{j=1}^{R'}}\, ;
\end{align}
\item  if $R'=R+ 1$, one has similarly
\begin{multline}
\bra{Q_\tau}\,\s_n^-\,\ket{Q_{\tau'}}
   = 2^{N-2R'}\, (-1)^{N-1}\,
   \frac{Q_{\tau'}(\xi_n-\eta)}{Q_\tau(\xi_n-\eta)}\\
   \times
   \frac{\pl_{k=1}^R a_n(\la_k) \pl_{j=1}^{R'}d_n(\la'_j)}{V(\la_R,\dots\la_1)\, V(\la'_1,\dots\la'_{R'})}\
  \det_R\mathcal{F}^{-}(\{\la'\},\{\la\},\xi_n)\vphantom{ \pl_{j=1}^{R'}}\, ;
\end{multline}
\item finally, if $R=R'$, the matrix element $\bra{Q_\tau}\,\s_n^-\,\ket{Q_{\tau'}}$ takes the form
\begin{multline}
\bra{Q_\tau}\,\s_n^-\,\ket{Q_{\tau'}}
   = 2^{N-2R-1}\, (-1)^{N-R}\,
     \frac{Q_\tau(\xi_n)}{Q_{\tau'}(\xi_n)}\,
     \frac{\pl_{k=1}^R a_n(\la_k) \pl_{j=1}^{R'}d_n(\la'_j)}{V(\la_1,\dots\la_R)\, V(\la'_{R'},\dots\la'_1)}
     \\
     \times
     \det_R\left[\mathcal{F}(\{\la\},\{\la'\})+\mathcal{P}^{(n)}(\{\la'\},\{\la\},\xi_n)\right],
     \vphantom{ \pl_{j=1}^{R'}}
\label{ff=}
\end{multline}
where the elements of the matrix $\mathcal{F}(\{\la\},\{\la'\})$ are
\begin{equation}
\mathcal{F}_{j,k}
=\frac{a(\la'_k)}{d(\la'_k)}\, Q_\tau(\la'_k-\eta)\, t(\la_j-\la'_k)+Q_\tau(\la'_k+\eta)\, t(\la'_k-\la_j),
\end{equation}
and where $\mathcal{P}^{(n)}(\{\la'\},\{\la\},\xi_n)$ is a rank one matrix,
\begin{equation}
     \mathcal{P}^{(n)}_{j,k}
     =\left(\frac{a(\la'_k)}{d(\la'_k)}\, Q_\tau(\la'_k-\eta)+Q_\tau(\la'_k+\eta)\right) t(\la_j-\xi_n).
\end{equation}
\end{itemize}
\end{theorem}

\begin{proof}
We consider the case $R\ge R'$ (the case with $R<R'$ can be done in a similar way). 
Using the reconstruction formula (\ref{inverse_problem}) and the fact that $\bra{Q_\tau}$ and $\ket{Q_{\tau'}}$ are eigenstates of the transfer matrix with respective eigenvalues $\tau(\lambda)$ and $\tau'(\lambda)$,  we immediately get the following representation for the matrix element $\bra{Q_\tau}\, \s_n^-\, \ket{Q_{\tau'}}$:
\begin{equation}
   \bra{Q_\tau}\, \s_n^-\, \ket{Q_{\tau'}}
   =(-1)^N\, \frac{\pl_{j=1}^{n-1}\tau(\xi_j)\pl_{j=n+1}^{N}\tau'(\xi_j)}{\pl_{j=1}^Na(\xi_j)}\,
   \bra{Q_\tau}\, D(\xi_n)\, \ket{Q_{\tau'}}.
\end{equation}
We can now use the ABA  type representations (\ref{ABA_like}) from which we straightforwardly obtain that 
\begin{equation}
   D(\xi_n)\, \ket{Q_{\tau'}}=(-1)^N\ket{\a},
\end{equation}
with $\a(\la)$ being a polynomial of degree $R'+1$ constructed from $Q_{\tau'}(\lambda)$ as
\begin{equation}
    \a(\la)=(\la-\xi_n)\, Q_{\tau'}(\la).
\end{equation}
Hence, it reduces the computation of the form factor $\bra{Q_\tau}\, \s_n^-\, \ket{Q_{\tau'}}$ to the computation of the scalar product $\moy{Q_\tau\, |\, \alpha}$ between the eigenstate $\bra{Q_\tau}$ and the separate state $\ket{\alpha}$. We can therefore directly use the results of Theorem~\ref{th-sc-pdt}. 
More precisely, in the case $R>R'+1$, we can use the result for the scalar product (\ref{sp_less}) to show that the corresponding form factor is zero. If $R=R'+1$ we apply (\ref{sp_slavnov}) which leads to the result (\ref{ff_+1}). Finally, in the case $R=R'$, we can rewrite the generalised Slavnov determinant (\ref{sp_more}) as a determinant of a sum of two $R\times R$ matrices (\ref{ff=}).
\end{proof}

We can now use the symmetries \eqref{x-symmetries} of the transfer matrix  to compute the matrix elements of the local operators $\s_n^+$ and $\s_n^z$ from the ones of $\sigma_n^-$. 

\begin{theorem}\label{th-ffz+}
Let $\ket{Q_\tau}$ and $\ket{Q_{\tau'}}$ be two  eigenstates of the transfer matrix constructed from  polynomials $Q_\tau(\lambda)$ and $Q_{\tau'}(\lambda)$ with respective degree $R$ and $R'$.
The matrix elements $\bra{Q_\tau}\,\s_n^z\,\ket{Q_{\tau'}}$ and $\bra{Q_\tau}\,\s_n^+\,\ket{Q_{\tau'}}$ of $\sigma_n^z$ and $\sigma_n^+$ are given in terms of $\bra{Q_\tau}\,\s_n^-\,\ket{Q_{\tau'}}$ as
\begin{align}
        &\bra{Q_\tau}\,\s_n^z\,\ket{Q_{\tau'}}=2(R'-R)\, \bra{Q_\tau}\,\s_n^-\,\ket{Q_{\tau'}},\\
        &\bra{Q_\tau}\,\s_n^+\,\ket{Q_{\tau'}}=(-1)^{R-R'}\, \bra{Q_\tau}\,\s_n^-\,\ket{Q_{\tau'}}.
\end{align}
\end{theorem}

\begin{proof}
The statement of the theorem follows from the fact that we can simply obtain $\sigma_n^z$ and $\sigma_n^+$ from $\sigma_n^-$ by means of the operators $S^x$ and $\Gamma^x$ \eqref{x-symmetries1}:
\begin{equation}
     \s_n^z=[S^x, \s_n^-],\qquad \s_n^+=\Gamma^x\,\s_n^-\,\Gamma^x.
\end{equation}
Hence, to obtain the matrix elements of $\sigma_n^z$ and $\sigma_n^+$ in terms of those of $\sigma_n^-$, we just have to compute the action of the operators $S^x$ and $\Gamma^x$ on the corresponding eigenstates. Using the following limit
\begin{equation}
   S^x=\lim_{\la\rightarrow\infty}\frac{\mathcal{T}(\la)}{\eta\,\la^{N-1}},
\end{equation}
we obtain
\begin{equation}
S^x \ket{Q_\tau}=(N-2R)\ket{Q_\tau}.
\end{equation}
The action of the operator $\Gamma^x$ on the eigenstates of the transfer matrix can be easily obtained from the relation
\begin{equation}
    \Gamma^x=(-i)^N\exp\left(\frac{i\pi}2 S^x\right).
\end{equation}
\end{proof}

We would like to mention to conclude this section that it is also possible to compute in this framework the matrix elements of more general quasi-local operators, i.e of operators obtained as a combination of local spin operators acting on  a finite number of sites of the lattice.

\section{Correspondence with the results obtained by Algebraic Bethe Ansatz}

The XXX chain is a unique example for which  the results obtained by SoV can be directly compared with the results obtained by ABA.
 The $SU(2)$ symmetry of the XXX monodromy matrix provides indeed a correspondence between the chain with antiperiodic boundary conditions  and the chain with the following twisted boundary conditions:
 \begin{equation}
       \s_{N+1}^a=\s^z_1\,\s^a_1\,\s^z_1,\quad a=x,y,z.
 \end{equation}
 The eigenstates of this twisted XXX chain can be constructed, in the ABA framework, as the eigenstates of the corresponding twisted transfer matrix,
 \begin{equation}\label{twisted-T}
 \mathcal{T}_-(\la)=\tr_0 \big[\sigma_0^z\, T_0(\la)\big]=A(\la)-D(\la),
 \end{equation}
 by using as usual the operators $B(\la)$ (or $C(\la)$) as creation operators on the reference state \eqref{ref-state}. More precisely, the state
 \begin{equation}
 \ket{\Psi(\{\la\})}=\pl_{j=1}^R B(\la_j)\ket{0}
 \end{equation} 
 is an eigenstate of the transfer matrix $ \mathcal{T}_-(\la)$ if and only if  the set of parameters $\{\lambda\}\equiv\{\la_1,\dots,\la_R\}$ satisfies the Bethe equations
 \begin{equation}
\pl_{a=1}^{ {R}}\frac{\lambda _{b}-\lambda _{a}+\eta }{\lambda
_{b}-\lambda _{a}-\eta }=\frac{a(\lambda _{b})}{d(\lambda _{b})},
\label{Bethe-XXX}
\end{equation}
and the corresponding eigenvalue is 
\begin{equation}
\tau_-(\la)=a(\la)\pl_{a=1}^{R}\frac{\lambda -\lambda _{a}-\eta }{\lambda
_{b}-\lambda _{a} }-d(\la)\pl_{a=1}^{R}\frac{\lambda -\lambda _{a}+\eta }{\lambda
_{b}-\lambda _{a} }.
\end{equation}
Let us  mention here that there exists an alternative way to construct the eigenstates of \eqref{twisted-T}, starting instead from the second reference state
\begin{equation}
 \ket{0'}= \mathop{\otimes}\limits_{n=1}^N  \left(\begin{array}{c} 0\\ 1\end{array}\right)_{\! [n]},
\end{equation} 
and using the operators $C(\la)$. Then, if $\{\la_1,\dots,\la_R\}$ satisfies the Bethe equations \eqref{Bethe-XXX}, the state
 \begin{equation}\label{Psi-hat}
 \ket{\widehat\Psi(\{\la\})}=\pl_{j=1}^R C(\la_j)\ket{0'},
 \end{equation} 
 is an eigenstate of the transfer matrix corresponding to  the eigenvalue
 \begin{equation}
     \widehat{\tau}_-(\la)=-\tau_-(\la).
\end{equation}
Note that the  Bethe equations \eqref{Bethe-XXX} issued from the ABA study of the twisted chain coincide with the Bethe equations \eqref{Bethe_anti} issued from the SoV study of the antiperiodic  chain, so that the completeness of the former in the ABA framework can be derived from the completeness of the latter. It is moreover possible to establish an explicit one-to-one correspondence between the eigenstates. 

To this aim we use the $SU(2)$ symmetry of the XXX monodromy matrix $T_0(\lambda)$: for any $U\in SU(2)$,  it is easy to see that
\begin{equation}
    [U_0\,\Gamma_{\!_U}, T_0(\la)]=0,\qquad\text{where}\quad\Gamma_{\!_U}=\mathop{\otimes}_{n=1}^N U_{n}.
\end{equation}
To establish the relation between two transfer matrices, we consider the following unitary matrix:
\begin{equation}\label{choice-U}
U=\frac1{\sqrt{2}}\left(\begin{array}{cc} \!\!\hphantom{-}1&1\\ \!\!-1&1\end{array}\right).
\end{equation} 
Using the property 
$U\s^x=\s^z U$ it is easy to establish  the following similarity transformation between the transfer matrices of two models:
\begin{equation}
\label{similarity}
\mathcal{T}_-(\la)=\Gamma_{\!_U}\,\mathcal{T}(\la)\, \Gamma_{\!_U}^{-1},
\end{equation}
This similarity evidently means that these two transfer matrices are isospectral, and that their eigenvectors are in one-to-one correspondence by means of the similarity matrix $\Gamma_{\!_U}$.
More precisely, we have the following result.

\begin{theorem}\label{th-corresp}
Let $\{\la\}=\{\la_1,\dots\la_R\}$ be a solution of the Bethe equations \eqref{Bethe-XXX}, and let us define the following polynomial of degree $R$ with roots $\la_j$, $1\le j\le R$:
\begin{equation}\label{Q-la}
    Q(\la)=\pl_{j=1}^R(\la-\la_j).
\end{equation}
Then the separate state $\ket{Q}$ is an eigenstate of the antiperiodic transfer matrix \eqref{transfer}
with  eigenvalue 
\begin{equation}\label{tau-la}
   \tau(\lambda)=-a(\lambda)\, \frac{Q(\lambda-\eta)}{Q(\lambda)}+d(\lambda)\, \frac{Q(\lambda+\eta)}{Q(\lambda)},
\end{equation}
whereas the Bethe vector $\ket{\widehat\Psi(\{\la\})}$ \eqref{Psi-hat} is an eigenstate of the twisted transfer matrix \eqref{twisted-T} with the same eigenvalue \eqref{tau-la}.
These two states are related by means of the similarity matrix constructed from \eqref{choice-U} as
\begin{equation}
      \ket{Q}=(-1)^{N(R-1)}\,2^{\frac N2-R}\, \Gamma_{\!_U}^{-1}\ket{\widehat{\Psi}(\{\la\})}.
\end{equation}
\end{theorem}

\begin{proof}
For any solution $\{\la\}$ of the Bethe equations we can construct such states $\ket{Q}$ and $\ket{\widehat{\Psi}(\{\la\})}$. Let us first show that the state   $\ket{\widehat{\Psi}(\{\la\})}$ is non-trivial.
The norm of $\ket{Q}$ is given by (\ref{gaudin}) while for $\ket{\widehat{\Psi}(\{\la\})}$ one can use the usual Algebraic Bethe Ansatz arguments to obtain a representation in terms of the same Gaudin determinant:
 \begin{equation}
\langle \, \widehat{\Psi}(\{\la\})\, \ket{\widehat{\Psi}(\{\la\})}=  \left(\pl_{n=1}^R d(\la_n)\right)^2\frac{\pl_{m,n=1}^R(\la_m-\la_n+\eta)}
{\pl_{m\neq n}(\la_m-\la_n)}\,\det_R \Phi_\tau.
\end{equation}
 Since $\ket{Q}$ is a nontrivial state by construction it follows that the corresponding Gaudin determinant is non zero and therefore that $\ket{\widehat{\Psi}(\{\la\})}$ is non-trivial.
  
  The similarity (\ref{similarity}) implies that 
\begin{equation}
  \mathcal{T}(\la)\, \Gamma_{\!_U}^{-1} \ket{\widehat{\Psi}(\{\la\})}
  =\tau(\lambda)\, \Gamma_{\!_U}^{-1} \ket{\widehat{\Psi}(\{\la\})},
\end{equation}
so that $\Gamma_{\!_U}^{-1} \ket{\widehat{\Psi}(\{\la\})}$ is an eigenstate of the antiperiodic transfer matrix $\mathcal{T}(\la)$ with the same eigenvalue $\tau(\lambda)$ \eqref{tau-la} as $\ket{Q}$.
Hence, since the spectrum of $\mathcal{T}(\la)$ is simple, this state should be proportional to $\ket{Q}$. The proportionality coefficient can be computed up to a phase factor from the Gaudin formula.
  
  It remains to compute this phase factor.
  Let us consider  the simplest separate state $\ket{1}$. It is an eigenstate of $\mathcal{T}(\lambda)$ with eigenvalue $d(\la)-a(\la)$, and  should therefore be proportional  to $\Gamma_{\!_U}^{-1}\ket{0'}$. It is easy to see (directly from the construction) that
\begin{equation}
  \langle\,0\, \ket{1}=1,
  \qquad 
  \bra{0}\, \Gamma_{\!_U}^{-1}\ket{0'}=\left(\frac{-1}{\sqrt{2}}\right)^N.
\end{equation}
 It leads to a useful relation:
 \begin{equation}
   \label{1-explicit}
   \ket{1}=(-\sqrt{2})^N \, \Gamma_{\!_U}^{-1}\ket{0'}
   = \mathop{\otimes}\limits_{n=1}^N \left(\!\!\!\begin{array}{c} \hphantom{-}1\\ -1\end{array}\right)_{[n]}.
\end{equation}
Expressing  now $\ket{Q}$ by means of the ABA-type representation \eqref{Qtau-form1}, we obtain
\begin{equation}
   \Gamma_{\!_U}\ket{Q}
   =(-1)^{RN}\, \Gamma_{\!_U} \pl_{j=1}^R D(\la_j)\, \ket{1}
   =2^{\frac N2}(-1)^{N(R-1)}\pl_{j=1}^R \left(\Gamma_{\!_U}\,D(\la_j)\,\Gamma_{\!_U}^{-1}\right)\ket{0'}.
\end{equation}
It is easy to see that
\begin{equation}
  \Gamma_{\!_U}\,D(\la)\,\Gamma_{\!_U}^{-1}=\frac 12\Big(A(\la)+B(\la)+C(\la)+D(\la)\Big).
\end{equation}
Due to the proportionality between $\Gamma_{\!_U}\ket{Q}$ and $\ket{\widehat{\Psi}(\{\la\})}$ only the term containing $R$ operators $C(\la_j)$ produces a non-zero contribution and we obtain the final result for the phase factor.
\end{proof}

Hence, Theorem~\ref{th-corresp} relates the eigenstates of the antiperiodic transfer matrix, which were constructed as separate states in the SoV framework, with the on-shell Bethe states for the twisted transfer matrix constructed from ABA (note that Bethe equations are crucial here).
On the one hand, it provides an easy way to prove the completeness of the ABA construction. On the other hand, it gives explicit representations valid in the homogeneous limit for the separate states: in particular, the equation   (\ref{1-explicit}) shows that the state $\ket{1}$ does not depend on the inhomogeneities. A similar result can be obtained for the state $\ket{1_\mathrm{alt}}$,
\begin{equation}
\ket{1_\mathrm{alt}}=\mathop{\otimes}\limits_{n=1}^N \left(\begin{array}{c} 1\\ 1\end{array}\right)_{\![n]}.
\end{equation}

This explicit correspondence between these two families of eigenstates can also be used to compare the expressions we have just obtained for the form factors of the antiperiodic model in the SoV framework with those that can be computed for the twisted model in the ABA framework.
Namely, let us fix two given eigenvalues $\tau(\lambda)$ and $\tau'(\lambda)$ of the antiperiodic (or twisted) transfer matrix, corresponding to two sets of solutions $\{\lambda\}\equiv\{\lambda_1,\ldots,\lambda_R\}$ and $\{\la'\}\equiv\{\la'_1,\ldots,\la'_{R'}\}$ of the Bethe equations \eqref{Bethe_anti} with associated Baxter polynomials $Q_\tau(\lambda)$ and $Q_{\tau'}(\lambda)$.
We shall denote as usual the corresponding $\mathcal{T}(\lambda)$-eigenstates by $\ket{Q_\tau}$ and $\ket{Q_{\tau'}}$, and the corresponding $\mathcal{T}_-(\lambda)$-eigenstates of the form \eqref{Psi-hat} by $\ket{\widehat{\Psi}(\{\la\})}$ and $\ket{\widehat{\Psi}(\{\la'\})}$.
Then, Theorem~\ref{th-corresp} implies the following relation between the form factors:
\begin{align}
  \bra{ Q_\tau}\, \sigma _{n}^{-}\,\ket{Q_{\tau^\prime } }
   &=   \bra{\widehat{\Psi}(\{\la\})}\, \sigma _{n}^{z}\, \ket{\widehat{\Psi}(\{\la'\})}
          -\frac{ \bra{\widehat{\Psi}(\{\la\})}\, \sigma _{n}^-\, \ket{\widehat{\Psi}(\{\la'\})} }{2i}
          \nonumber\\
   &
      \hspace{4.35cm}
          +\frac{\bra{\widehat{\Psi}(\{\la\})}\, \sigma _{n}^+\, \ket{\widehat{\Psi}(\{\la'\})} }{2i},
\label{s_-to-s_z-s_pm-1} \\
  &=\left\{ 
       \begin{aligned}
   &- \bra{\widehat{\Psi}(\{\la\})}\, \sigma _{n}^-\, \ket{\widehat{\Psi}(\{\la'\})} /2i
   &\text{ \ for }  &R=R^{\prime }+1,
   \\ 
   &\bra{\widehat{\Psi}(\{\la\})}\, \sigma _{n}^{z}\, \ket{\widehat{\Psi}(\{\la'\})} 
  &\text{ \ for }  &R=R^{\prime },
   \\ 
   &\bra{\widehat{\Psi}(\{\la\})}\, \sigma _{n}^+\, \ket{\widehat{\Psi}(\{\la'\})} /2i
  &\text{ \ for }   &R^{\prime }=R+1.
  \end{aligned}
  \right. 
    \label{s_-to-s_z-s_pm-2}
\end{align}
Since we are now able to compute these matrix elements independently  in the SoV and in the ABA framework, and since the obtained representations may appear slightly different, it is important to check their effective  coincidence. This is explicitly done in Appendix~\ref{ap-check-ff}.

\section*{Acknowledgements}
J. M. M., G. N. and V. T. are supported by CNRS. N. K, J. M. M. and V. T. are supported by ANR grant ``DIADEMS''. N. K. is supported by the BQR program of the Université de Bourgogne and would like to thank LPTHE, University Paris VI, Laboratoire de Physique, ENS-Lyon and the Galileo Galilei Institute in Florence for hospitality.

\appendix

\section{Proof of Identity~\ref{gen_Slavnov_relation_th}}

In this appendix we prove the identity~\ref{gen_Slavnov_relation_th} concerning the equality between the quantity $\mathcal{A}^-_{\{x\}\cup\{y\}}\big[\mu E^+_{\{\xi\}}\big]$ and the generalized Slavnov determinant \eqref{Slavnov-gen} when $\{x\}$ is a solution to the Bethe equations \eqref{Bethe_equation}. We shall proceed in a way similar to the proof of the identity~\ref{Slavnov_relation_th}.

We consider a set of pairwise distinct Bethe roots $\{x\}\equiv\{x_1,\dots , x_M\}$, and a set of arbitrary complex numbers $\{y\}\equiv\{y_1,\dots, y_{M+S}\}$. As previously we introduce  the polynomial $Q(\la)$ with roots $x_1,\ldots,x_M$:
\begin{equation}
     Q(\la)=\pl_{j=1}^N(\la-x_j),
\end{equation}
and the $M$ polynomials $Q_k(\la)$ obtained from $Q(\la)$ as
\begin{equation}
      Q_k(\la)=\frac{Q(\la)}{(\la-x_k)}.
\end{equation}
We also introduce polynomials $Z(\la)$ of degree $S$ with generic roots $z_1,\dots,z_S$ such that $z_a\neq x_j$, $z_a\neq x_j-\eta$,
\begin{equation}
Z(\la)=\pl_{a=1}^S(\la-z_a), \qquad Z_k(\la)=\frac{Z(\la)}{(\la-z_k)}.
\end{equation}
As before we consider the auxiliary $(2M+S)\times (2M+S)$ matrix $\mathcal{C}$ composed of the coefficients of the following polynomials:
\begin{align}
Q_k(\la)\,Q_k(\la+\eta)&=\sul_{j=1}^{2M+S} \mathcal{C}_{k,j}\, \la^{j-1},\nonumber\\
Q_k(\la)\, Q(\la+\eta)&=\sul_{j=1}^{2M+S} \mathcal{C}_{k+M,j}\, \la^{j-1},\nonumber\\
Q(\la)\, Q(\la+\eta)\, Z_k(\la)&=\sul_{j=1}^{2M+S} \mathcal{C}_{k+2M,j}\, \la^{j-1}.
\end{align}
Evidently the matrix elements $ \mathcal{C}_{k,j}$ are zero for $k\le 2M$ and $j>2M$. However the matrix $\mathcal{C}$ is invertible and its determinant can easily be computed using the following identities:
\begin{align*}
 &\sul_{a=1}^{2M+S}\mathcal{C}_{b,a}\, (x_k-\eta)^{a-1}=\delta_{b,k} \, Q_k (x_k)\, Q_k (x_k-\eta),
 \\
 &\sul_{a=1}^{2M+S}\mathcal{C}_{b,a}\, x_k^{a-1}= \delta_{b,k+M}\, Q(x_k+\eta)\, Q_k (x_k)
 +\delta_{b,k} \, Q_k (x_k)\, Q_k (x_k+\eta),
 \\
 &\sul_{a=1}^{2M+S}\mathcal{C}_{b,a}\, z_k^{a-1}= \delta_{b,k+2M}\, Q(z_k+\eta)\,Q (z_k)\,Z_k(z_k)+\mathcal{A}_{b,k},
\end{align*}
with $\mathcal{A}_{b,k}=0$ if $b>2M$. It means that the product of the matrix $\mathcal{C}$ and the corresponding
Vandermonde matrix is triangular, so that the determinant can be easily computed:
\begin{multline}
V(\{x-\eta\}\cup\{x\}\cup\{z\})\ \det_{2M+S} \mathcal{C}\\
 =\pl_{k=1}^M \left[ Q^2_k (x_k)\, Q_k (x_k-\eta)\, Q (x_k+\eta)\right]
    \pl_{j=1}^S \left[ Q(z_j)\, Q(z_j+\eta)\, Z_j(z_j)\right] .
\end{multline}
Hence,
\begin{equation}
\det_{2M+S}  \mathcal{C}= (-1)^{\frac{M(M+1)+S(S-1)}{2}}\, V(\{z\})
 \pl_{k=1}^M \big[ Q^2_k (x_k)Q_k (x_k+\eta) \big]. 
\end{equation}
Let us now  compute the following product of two determinants:
 \begin{equation}
 \mathcal{A}^-_{\{x\}\cup\{y\}}\!\left[\mu E^+_{\{\xi\}}\right]\ 
 \det_{2M+S} \mathcal{C}
 =\frac{\det_{2M+S}\left(
 \begin{array}{ccc}
 \mathcal{G}^{(1,1)}&\mathcal{G}^{(1,2)}\\
 \mathcal{G}^{(2,1)}&\mathcal{G}^{(2,2)}\\
 \mathcal{G}^{(3,1)}&\mathcal{G}^{(3,2)}
 \end{array}\right)}{V(\{x\}\cup\{y\})}.
 \end{equation}
 As in the case $S=0$, the $M\times M$ block $ \mathcal{G}^{(1,1)}$ is given by
\begin{align}
  \mathcal{G}^{(1,1)}_{j,k}
  &= \sul_{a=1}^{2M+S}\mathcal{C}_{j,a}\left(x_k^{a-1}-\mu  E^+_{\{\xi\}}(x_k)\, (x_k-\eta)^{a-1}\right)
        \nonumber\\
  &= Q_j(x_k)\, Q_j(x_k+\eta)-\mu E^+_{\{\xi\}}(x_k)\, Q_j(x_k)\, Q_j(x_k-\eta)
  \nonumber\\
  &=\delta_{jk}\,Q_k(x_k)\left(Q_k(x_k+\eta)-\mu E^+_{\{\xi\}}(x_k)\, Q_k(x_k-\eta)\right).
\end{align}
Due to the Bethe equations $ \mathcal{G}^{(1,1)}_{j,k}=0$.
Another block which turns out to be trivial is the  $S\times M$ matrix $ \mathcal{G}^{(3,1)}$:
\begin{align}
  \mathcal{G}^{(3,1)}_{j,k}
  &=\sul_{a=1}^{2M+S}\mathcal{C}_{j+2M,a}\left(x_k^{a-1}-\mu  E^+_{\{\xi\}}(x_k)\, (x_k-\eta)^{a-1}\right)
       \nonumber \\
   &= Q(x_k)\, Q(x_k+\eta)\, Z_j(x_k)-\mu E^+_{\{\xi\}}(x_k)\, Q(x_k)\, Q(x_k-\eta)\, Z_j(x_k-\eta)
       \nonumber\\
    &=0,
\end{align}
due to the zero factor $Q(x_k)$ in both terms.
It means that there is no need to compute the block   $ \mathcal{G}^{(2,2)}$ as it does not contribute to the determinant. The $M\times M$ block   $ \mathcal{G}^{(2,1)}$ is diagonal:
\begin{align}
   \mathcal{G}^{(2,1)}_{j,k}
   &=\sul_{a=1}^{2M+S}\mathcal{C}_{j+M,a}\left(x_k^{a-1}-\mu  E^+_{\{\xi\}}(x_k) (x_k-\eta)^{a-1}\right)\nonumber\\
   &=Q_j(x_k)\, Q(x_k+\eta)-\mu E^+_{\{\xi\}}(x_k)\, Q(x_k)\, Q_j(x_k-\eta)\nonumber\\
   &=\delta_{jk}\, Q_k(x_k)\,Q(x_k+\eta).
\end{align}
The non-trivial part of this product is contained in the two remaining blocks. The $(M+S)\times M$ block $\mathcal{G}^{(1,2)}$ has a usual form of a Slavnov matrix:
\begin{align}
     \mathcal{G}^{(1,2)}_{j,k}
     &=\sul_{a=1}^{2M+S}\mathcal{C}_{j,a}\left(y_k^{a-1}-\mu  E^+_{\{\xi\}}(y_k)\, (y_k-\eta)^{a-1}\right)\nonumber\\ 
    &=Q_j(y_k)\, Q_j(y_k+\eta)-\mu E^+_{\{\xi\}}(y_k)\, Q_j(y_k)\, Q_j(y_k-\eta)\nonumber\\
    &=-\frac 1\eta \, Q(y_k)\, Q(y_k-\eta)\,  \widetilde{\mathcal {H}}_{j k},
\end{align}
while the $(M+S)\times S$ block  $\mathcal{G}^{(3,2)}$ is
\begin{align}
     \mathcal{G}^{(3,2)}_{j,k}
     &=\sul_{a=1}^{2M+S}\mathcal{C}_{j+2M,a}\left(y_k^{a-1}-\mu  E^+_{\{\xi\}}(y_k) \, (y_k-\eta)^{a-1}\right)\nonumber\\ 
    &=Q(y_k)\left(Q(y_k+\eta)\, Z_j(y_k)-\mu E^+_{\{\xi\}}(y_k)\, Q(y_k-\eta)\, Z_j(y_k-\eta)\right).
\end{align}
Now to get rid of the arbitrary polynomials $Z(\la)$ we introduce the following $S\times S$ matrix
$\widetilde{\mathcal{C}}$ defined as
\begin{equation}
Z_j(\la)=\sul_{a=1}^S \widetilde{\mathcal{C}}_{j,a}\, \la^{a-1}.
\end{equation}
It is easy to see that 
\begin{equation}
\det_S \widetilde{\mathcal{C}}=(-1)^{\frac{S(S-1)}2}\, V(\{z\}).
\end{equation}
Together with the following representation for the block $\mathcal{G}^{(3,2)}$,
\begin{multline}
Q(y_k+\eta)\, Z_j(y_k)-\mu E^+_{\{\xi\}}(y_k)\, Q(y_k-\eta)\, Z_j(y_k-\eta)\\
 =\sul_{a=1}^S \widetilde{\mathcal{C}}_{j,a}\left(  Q(y_k+\eta)\, y_k^{a-1}-\mu E^+_{\{\xi\}}(y_k)\, Q(y_k-\eta)(y_k-\eta)^{a-1}\right),
\end{multline}
it leads to the  expression  (\ref{gen_Slavnov_relation}).

\section{Explicit comparison of the form factor representations issued from SoV and from ABA}
\label{ap-check-ff}

In this appendix we explicitly check that the expressions of the form factors we have obtained from our SoV study of the XXX antiperiodic spin chain are consistent with the one issued from the ABA study \cite{KitMT99} of the twisted spin chain, i.e. that the relation \eqref{s_-to-s_z-s_pm-2} is effectively satisfied.

The fact that the SoV and the ABA expressions for the form factors coincide if $|R-R^{\prime }|=1$ is quite simple to prove. We shall therefore provide here the details for the verification of the last case $R=R^{\prime }$ only.

The direct computation of the form factor $\bra{\widehat{\Psi}(\{\la\})}\, \sigma _{n}^{z}\, \ket{\widehat{\Psi}(\{\la'\})} $ by ABA leads to the following expression:
\begin{multline}\label{expr-ABA}
   \bra{\widehat{\Psi}(\{\la\})}\, \sigma _{n}^{z}\, \ket{\widehat{\Psi}(\{\la'\})} 
    =\mathcal{S}_{R}^{(-1)}(\{\lambda \} ,\{ \lambda ^{\prime }\} |\{ \xi \})
    \\
    +2\sum_{m=1}^{R}
    \frac{Q_{\tau^\prime }(\lambda _{m}^{\prime }-\eta )}{Q_\tau(\lambda_{m}^{\prime }-\eta )}\,
    \mathcal{S}_{R}^{(-1,m)}(\{ \lambda \},\{ \lambda ^{\prime }\} |\{ \xi \} |\xi_n),
\end{multline}
where $\mathcal{S}_{R}^{(-1,m)}(\{ \lambda \},\{ \lambda ^{\prime }\} |\{ \xi \} |\xi_n)$ is obtained from $\mathcal{S}_{R}^{(-1)}(\{\lambda \} ,\{ \lambda ^{\prime }\} |\{ \xi \})$ \eqref{Slavnov-det} by
substituting $\xi _{n}$ to $\lambda _{m}^{\prime }$ in the $m$-th column of the
matrix $\mathcal{H}^{(-1)}(\{\la\},\{\la'\}|\{\xi\})$.
The SoV computation gives instead:
\begin{multline}\label{expr-SoV}
    \bra{ Q_\tau}\, \sigma _{n}^{-}\,\ket{Q_{\tau^\prime } }
   = \mathcal{S}_{R}^{(-1)}(\{\lambda \} ,\{ \lambda ^{\prime }\} |\{ \xi \})
    \\
+\sum_{m=1}^{R}
  \frac{a(\lambda _{m}^{\prime})\, Q_\tau(\lambda _{m}^{\prime }-\eta )
         +d(\lambda _{m}^{\prime })\, Q_\tau (\lambda_{m}^{\prime }+\eta )}
         {a(\lambda _{m}^{\prime })\,
         Q_\tau(\lambda _{m}^{\prime}-\eta )}\,
         \mathcal{S}_{R}^{(-1,m)}(\{ \lambda \},\{ \lambda ^{\prime }\} |\{ \xi \}|\xi_n ).
\end{multline}

In the case $\tau (\lambda )=\tau ^{\prime }(\lambda )$, the equality between \eqref{expr-ABA} and \eqref{expr-SoV} is a consequence of the following identity:
\begin{align}
  \frac{a(\lambda _{m})\,Q_\tau(\lambda _{m}-\eta )+d(\lambda _{m})\,Q_\tau(\lambda_{m}+\eta )}{a(\lambda _{m})\,Q_\tau(\lambda _{m}-\eta )}
     &=\frac{2a(\lambda_{m})\, Q_\tau(\lambda _{m}-\eta )+\tau (\lambda _{m})\, Q_\tau(\lambda _{m})}{a(\lambda_{m})\,Q_\tau(\lambda _{m}-\eta )}\nonumber\\
     &=2.
\end{align}

Let us now consider the case $R=R^{\prime }$ for $\tau (\lambda )\neq \tau ^{\prime }(\lambda)$.
We therefore want to compute the difference of the two expressions \eqref{expr-ABA} and \eqref{expr-SoV},
\begin{multline}\label{diff-1}
 \bra{ Q_\tau}\, \sigma _{n}^{-}\,\ket{Q_{\tau^\prime } } 
-\bra{\widehat{\Psi}(\{\la\})}\, \sigma _{n}^{z}\, \ket{\widehat{\Psi}(\{\la'\})}   \\
 =\sum_{m=1}^{R}
 \frac{a(\lambda _{m}^{\prime })
          \big[2Q_{\tau^\prime }(\lambda_{m}^{\prime }-\eta )-Q_\tau(\lambda _{m}^{\prime }-\eta )\big]
         -d(\lambda_{m}^{\prime })\,Q_\tau(\lambda _{m}^{\prime }+\eta )}
        {a(\lambda _{m}^{\prime})\, Q_\tau(\lambda _{m}^{\prime }-\eta )}\\
        \times
 \mathcal{S}_{R}^{(-1,m)}(\{ \lambda \},\{ \lambda ^{\prime }\} |\{ \xi \}|\xi_n ),
\end{multline}
and show that it vanishes.

One can first notice that the right hand side of \eqref{diff-1} can be seen as the development of the  determinant of a larger matrix $\widehat{\mathcal{S}}^{(n)}$ with one more line and column:
\begin{multline}
  \bra{ Q_\tau}\, \sigma _{n}^{-}\,\ket{Q_{\tau^\prime } } 
  -\bra{\widehat{\Psi}(\{\la\})}\, \sigma _{n}^{z}\, \ket{\widehat{\Psi}(\{\la'\})}  \\
  =\frac{\det_{R+1} \widehat{\mathcal{S}}^{(n)}}
    {V(\lambda _{1},\ldots ,\lambda _{R})\, 
     V(\lambda _{1}^{\prime},\ldots,\lambda _{R}^{\prime })
     \prod_{a=1}^{R}Q_\tau(\lambda _{a}^{\prime })},
\end{multline}
where $\widehat{\mathcal{S}}^{(n)}$ is the $(R+1)\times(R+1)$ matrix of elements, for $1\le j,k\le R$:
\begin{align*}
   &\big[ \widehat{\mathcal{S}}^{(n)} \big]_{j,k}
   =Q_{\tau,j}(\lambda'_k)\left( Q_{\tau,j}(\lambda'_k-\eta) +\frac{d(\lambda'_k)}{a(\lambda'_k)}\, Q_{\tau,j}(\lambda'_k+\eta)\right),
   \\
   &\big[ \widehat{\mathcal{S}}^{(n)} \big]_{j,R+1}=Q_{\tau,j}(\xi_n)\, Q_{\tau,j}(\xi_n-\eta),
   \\
   &\big[ \widehat{\mathcal{S}}^{(n)} \big]_{R+1,k}
   =Q_\tau(\lambda_k^{\prime })
     \left( 2Q_{\tau^\prime }(\lambda _k^{\prime }-\eta)-Q_\tau(\lambda _k^{\prime }-\eta )
            -\frac{d(\lambda _k^{\prime })}{a(\lambda_k^{\prime })}\, Q_\tau(\lambda_k^{\prime }+\eta )\right),
    \\
   &\big[ \widehat{S}^{(n)} \big]_{R+1,R+1}=\gamma,
\end{align*}
with $\gamma$ being an arbitrary complex number. Here we have used the shorthand notation
\begin{equation}
  Q_{\tau,j}(\lambda)=\frac{Q_\tau(\lambda)}{\lambda-\lambda_j}.
\end{equation}

We now want to rewrite this determinant in terms of the determinant of a $(2R+1)\times (2R+1)$ matrix $\mathcal{G}=\mathcal{C}\cdot\mathcal{X}$.
To this aim we introduce the matrices $\mathcal{C}$ and $\mathcal{X}$ as follows.

The $(2R+1)\times (2R+1)$ matrix $\mathcal{C}$ is defined by its elements $\mathcal{C}_{j,k}$, $1\le j,k \le  2R+1$ such that
\begin{align}
  &\sum_{j=1}^{2R+1}\lambda ^{j-1}\, \mathcal{C}_{m,j}
    =Q_{\tau, m}(\lambda )\, Q_{\tau, m}(\lambda -\eta ),
    \qquad
    \forall m\in \{ 1,\ldots,R\} ,
     \\ 
  &\sum_{j=1}^{2R+1}\lambda ^{j-1}\, \mathcal{C}_{m+R,j}
    =g_{m}(\lambda ),
    \qquad\forall m\in \{ 1,\ldots,R\} ,
     \\ 
  &\sum_{j=1}^{2R+1}\lambda ^{j-1}\, \mathcal{C}_{2R+1,j}
    =h(\lambda ),
\end{align}
where the functions $h(\lambda)$ and $g_m(\lambda)$, $1\le m\le R$, are given as
\begin{align}
    &h(\lambda ) =\big( Q_\tau(\lambda )-Q_{\tau^\prime }(\lambda )\big)
                              \big( Q_{\tau^\prime }(\lambda -\eta )-Q_\tau(\lambda -\eta )\big) 
                              +Q_{\tau^\prime }(\lambda)\, Q_{\tau^\prime }(\lambda -\eta ), \\
    &g_{m}(\lambda ) =\prod_{j\neq m,j=1}^{2R}(\lambda -g_{j}).
\end{align}
Here $g_{j}$ are general complex numbers which have only to satisfy the
condition:
\begin{equation}
\det_{1\le j,k\le R} \big[ G_{j}(\lambda _{k})\big]\neq 0,
\quad\text{where}\quad 
G_{k}(\lambda)=g_{k}(\lambda )+\frac{d(\lambda )}{a(\lambda )}\, g_{k}(\lambda +\eta ).
\end{equation}
In its turn, the $(2R+1)\times (2R+1)$ matrix $\mathcal{X}$ has for elements
\begin{align}
   &\mathcal{X}_{j,k}=T_{j}(\lambda _{k}^{\prime }), \\
   &\mathcal{X}_{j,R+1+k}=T_{j}(\lambda _{k}{})-\delta _{j,2R+1}\, H( \lambda_{k}) ,\\
   &\mathcal{X}_{j,R+1}=T_{j}(\xi _{n}{})+\delta _{j,2R+1}(\gamma-H(\xi _{n})),
\end{align}
for $k\in \{ 1,\dots,R \} $\ and $j\in \{ 1,\ldots,2R+1\} $,
where we have set
\begin{align}
   &T_{j}(\lambda )=\lambda ^{j-1}+\frac{d( \lambda ) }{a(\lambda ) }\, (\lambda +\eta )^{j-1},\\
   &H( \lambda )=h(\lambda )+\frac{d( \lambda ) }{a( \lambda ) }\, h(\lambda +\eta ).
\end{align}
We can write the matrix product of $\mathcal{C}$ and $\mathcal{X}$ in a block form
\begin{equation}
     \mathcal{C}\cdot \mathcal{X} =\mathcal{G}
     \equiv \begin{pmatrix}
     \mathcal{G}_{11} & \mathcal{G}_{12} & \mathcal{G}_{13}=0 \\ 
     \mathcal{G}_{21} & \mathcal{G}_{22} & \mathcal{G}_{23} \\ 
     \mathcal{G}_{31} & \mathcal{G}_{32} & \mathcal{G}_{33}=0
\end{pmatrix},
\end{equation}
where we have used that $\mathcal{C}_{m,2R+1}=0,$ for any $m\in \{ 1,\ldots,2R \}$.
Here $\mathcal{G}_{11}=\big(F_{j}(\lambda _{k}^{\prime })\big)$, $\mathcal{G}_{13}=\big(F_{j}(\lambda_{k})\big)$, $\mathcal{G}_{21}=\big(G_{j}(\lambda _{k}^{\prime })\big)$ and $\mathcal{G}_{23}=\big(G_{j}(\lambda _{k})\big)$ are $R\times R$ matrices;
$\mathcal{G}_{22}=\big(G_{j}(\xi _{n})\big)$ and  $\mathcal{G}_{12}=\big(F_{j}(\xi _{n})\big)$ are $R\times 1$ columns;
$\mathcal{G}_{31}=\big(K(\lambda _{k}^{\prime })\big)$ and $\mathcal{G}_{33}=\big(K(\lambda _{k}^{\prime })\big)$
are $1\times R$ rows; finally $\mathcal{G}_{32}=\gamma $. Here we have defined
\begin{align*}
    &F_{j}(\lambda ) 
    =Q_{\tau,j}(\lambda )\left( Q_{\tau,j}(\lambda -\eta )
        +\frac{d(\lambda )}{a(\lambda )}\,Q_{\tau,j}(\lambda +\eta )\right) , \\
     &K(\lambda )  = Q_\tau(\lambda )
     \left( Q_{\tau^\prime }(\lambda -\eta )-Q_\tau(\lambda-\eta )
     +\frac{d(\lambda )}{a(\lambda )}\, \big[Q_{\tau^\prime }(\lambda +\eta)-Q_\tau(\lambda +\eta )\big]\right) .
\end{align*}
Note that we have
\begin{equation}
H( \lambda _{m}^{\prime } ) =K(\lambda _{m}^{\prime })\quad\forall m\in \{ 1,\ldots,R \} .
\end{equation}
If we now use the block determinant formula to compute the determinant of
the matrix $\mathcal{G}$ we obtain
\begin{multline}
     \frac{\det_{2R+1}\mathcal{G}}{V(\lambda _{1},\ldots,\lambda_{R})\, V(\lambda _{1}^{\prime },\ldots,\lambda _{R}^{\prime})\prod_{a=1}^{R}Q_\tau(\lambda _{a}^{\prime })}
     \\
     =\left( \bra{ Q_\tau}\, \sigma _{n}^{-}\,\ket{Q_{\tau^\prime } } 
-\bra{\widehat{\Psi}(\{\la\})}\, \sigma _{n}^{z}\, \ket{\widehat{\Psi}(\{\la'\})} 
\right) \det_{1\le j,k\le R}\big[ G_{j}(\lambda _{k})\big],
\end{multline}
once we observe that,  thanks to the Bethe equations 
$$a(\lambda _{m}^{\prime })\,Q_{\tau^\prime }(\lambda_{m}^{\prime }-\eta )=d(\lambda _{m}^{\prime })\,Q_{\tau^\prime }(\lambda_{m}^{\prime }+\eta ),$$
 one has
\begin{equation}
     K(\lambda _{m}^{\prime })
     =Q_\tau(\lambda _{m}^{\prime })\left( 2Q_{\tau^\prime}(\lambda _{m}^{\prime }-\eta )
      -Q_\tau(\lambda _{m}^{\prime }-\eta )
      -\frac{d(\lambda _{m}^{\prime })}{a(\lambda_{m}^{\prime })}\, Q_\tau(\lambda _{m}^{\prime }+\eta )\right) .
\end{equation}

This leads to a representation of the quantity \eqref{diff-1} that we want to compute in terms of the product of the two determinants of $\mathcal{C}$ and $\mathcal{X}$:
\begin{multline}\label{blabla}
   \bra{ Q_\tau}\, \sigma _{n}^{-}\,\ket{Q_{\tau^\prime } } 
-\bra{\widehat{\Psi}(\{\la\})}\, \sigma _{n}^{z}\, \ket{\widehat{\Psi}(\{\la'\})} \\
 =\frac{\det_{2R+1}\mathcal{C} \ \det_{2R+1}\mathcal{X} }{V(\lambda _{1},\ldots,\lambda
_{R})\, V(\lambda _{1}^{\prime },\ldots,\lambda _{R}^{\prime})
  \prod_{a=1}^{R}Q_\tau(\lambda _{a}^{\prime })\ \det_{1\le j,k\le R}\big[ G_{j}(\lambda _{k})\big]}.
\end{multline}

We shall now show that
\begin{equation}
   \det_{2R+1}\mathcal{X}=0,
\end{equation}
which will conclude our proof.
In order to do so, let us introduce the $(2R+1)\times(2R+1)$ matrix $\mathcal{B}$ with elements given by the coefficients of the following polynomials:
\begin{align}
  &\sum_{j=1}^{2R+1}\lambda ^{j-1}\, \mathcal{B}_{m,j}
    =Q_{\tau, m}(\lambda )\, Q_{\tau, m}(\lambda -\eta ),
    \qquad
    \forall m\in \{ 1,\ldots,R\} ,
     \\ 
  &\sum_{j=1}^{2R+1}\lambda ^{j-1}\, \mathcal{B}_{m+R,j}
    =Q_{\tau, m}(\lambda )\, Q_{\tau}(\lambda -\eta ),
    \qquad\forall m\in \{ 1,\ldots,R\} ,
     \\ 
  &\sum_{j=1}^{2R+1}\lambda ^{j-1}\, \mathcal{B}_{2R,j} =Q_\tau(\lambda)\, Q_\tau(\lambda-\eta)+h(\lambda),
  \\
  &\sum_{j=1}^{2R+1}\lambda ^{j-1}\, \mathcal{B}_{2R+1,j} = h(\lambda).
\end{align}
Note that this matrix has a non zero determinant:
\begin{equation}
\det_{2R+1}\mathcal{B}
 =\frac{\left( -1\right) ^{R}}{\eta }\, h(\lambda _{R})\, V(\{\lambda _{i}\})\, V(\{\lambda _{i}+\eta\})\,
    \prod_{m=1}^{R-1}Q_{\tau, m}(\lambda _{m}-\eta )\neq 0.
\end{equation}
Now we can compute the matrix product
\begin{equation}
   \widetilde{\mathcal{G}}\equiv\mathcal{B}\cdot \mathcal{X}
   =\begin{pmatrix}
\mathcal{G}_{11} & \mathcal{G}_{12} & \mathcal{G}_{13}=0 \\ 
 \widetilde{\mathcal{G}}_{21} &  \widetilde{\mathcal{G}}_{22} &  \widetilde{\mathcal{G}}_{23} \\ 
\mathcal{G}_{31} & \mathcal{G}_{32} & \mathcal{G}_{33}=0
\end{pmatrix},
\end{equation}
where we need to precise only the $R\times R$ matrices $ \widetilde{\mathcal{G}}_{23}$:
\begin{equation*}
\left(  \widetilde{\mathcal{G}}_{23}\right) _{j,k}
=(1-\delta _{j,R})\, \delta _{j,k}\, Q_{\tau,j}(\lambda_{j})\, Q_{\tau, j}(\lambda _{j}-\eta ),
\quad\forall j,k\in \{1,\ldots,R\} .
\end{equation*}
Indeed, the $\mathcal{G}$ blocks are defined above, whereas the $R\times R$ matrix $ \widetilde{\mathcal{G}}_{21}$ and the $R\times 1$ column $ \widetilde{\mathcal{G}}_{22}$ have no influence on the computation of the
determinant of $ \widetilde{\mathcal{G}}$. It is then simple to show that, after some row and
column exchange, we can write a new $(2R+1)\times (2R+1)$ matrix such that
\begin{equation}
\det_{2R+1}  \widetilde{\mathcal{G}} =\det_{2R+1} \widehat{\mathcal{G}}
=\det_{R+2}\mathcal{A}\ \det_{R-1} \widehat{\mathcal{G}}_{23},
\quad \text{with} \quad 
\widehat{\mathcal{G}}= 
\begin{pmatrix}
\mathcal{A} & \mathcal{A}'=0 \\ 
\mathcal{A}'' & \widehat{\mathcal{G}}_{23}
\end{pmatrix}
\end{equation}
where $\mathcal{A}$ is a $(R+2)\times (R+2)$ matrix, $\mathcal{A}'$ is a $(R+2)\times (R-1)$ matrix, $\mathcal{A}''$
is a $(R-1)\times (R+2)$ matrix, and $\widehat{\mathcal{G}}_{23}$ is a $(R-1)\times (R-1)$ diagonal
invertible matrix. In particular, we have defined
\begin{align}
   &\big[ \widehat{\mathcal{G}}_{23}\big] _{j,k} 
   =\delta _{j,k}\, Q_{\tau, j}(\lambda_{j})\, Q_{\tau, j}(\lambda _{j}-\eta ),
   \quad\forall j,k\in \{1,\ldots,R-1\} , \\
  &\mathcal{A} 
  =
\begin{pmatrix}
   \mathcal{G}_{11} & \mathcal{G}_{12} & 0_{R\times 1} \\ 
   \mathcal{G}_{31} & \gamma  & 0 \\ 
\left( \widetilde{\mathcal{G}}_{21}\right) _{R} & \left( \widetilde{\mathcal{G}}_{22}\right) _{R} & 0
\end{pmatrix} ,
\end{align}
where $\left( \widetilde{\mathcal{G}}_{21}\right) _{R}$ and $\left( \widetilde{\mathcal{G}}_{22}\right) _{R}$ are
respectively the last row and the last element of $\widetilde{\mathcal{G}}_{21}$ and $\widetilde{\mathcal{G}}_{22}$.
This proves our statement as $\det_{R+2}\mathcal{A}=0$.



\end{document}